\documentclass{article}

\usepackage{arxiv}

\usepackage{amsmath, amssymb, amsthm, graphicx, booktabs, makecell}
\usepackage{tcolorbox}
\usepackage{subcaption}
\usepackage{multirow}
\usepackage{tikz}
\usetikzlibrary{positioning}
\usepackage{tabularx}

\usepackage{xcolor}
\usepackage[breaklinks]{hyperref}
\hypersetup{
  colorlinks = {true},
  citecolor = {blue},
  menubordercolor = {red},
  citebordercolor = {white},
  linkbordercolor = {red},
  urlcolor = {black},
  pdftitle = {The Innate Economic Preferences of Language Models},
  pdfauthor = {Joy Buchanan and Joshua Foster},
  pdfsubject = {Economic preferences of language models},
  pdfkeywords = {Revealed Preference, Structural Discrete Choice, AI Alignment}
}

\usepackage{fnpct}

\usepackage{natbib}

\renewcommand{\headeright}{Preprint}
\renewcommand{\undertitle}{Preprint}
\renewcommand{\shorttitle}{The Innate Economic Preferences of Language Models}

\theoremstyle{plain}
\newtheorem{proposition}{Proposition}[section]

\theoremstyle{definition}

\newtheorem{definition}{Definition}[section]
\newtheorem{measure}{Measure}[section]
\theoremstyle{remark}

\title{The Innate Economic Preferences\\ of Language Models}

\author{
  Joy Buchanan\thanks{Email: \href{mailto:jbuchan1@samford.edu}{jbuchan1@samford.edu}} \\
  Brock School of Business \\
  Samford University
  \And
  Joshua Foster\thanks{Email: \href{mailto:jfoster@ivey.ca}{jfoster@ivey.ca}; supported in part by funding from the Social Sciences and Humanities Research Council of Canada.} \\
  Ivey Business School \\
  Western University
}

\date{July 28, 2026}

\begin{document}

\maketitle

\begin{abstract}
  Language models increasingly settle real resource tradeoffs on behalf of principals yet their economic preferences remain unobserved. We demonstrate their generation rule is isomorphic to the random utility model of discrete choice. This allows internal logit scores to structurally identify preferences. Estimating risk attitudes across twelve models in a portfolio task reveals universal but heterogeneous risk aversion. Although models reject strictly dominated options, their elicited preferences fail invariance tests and violate the independence of irrelevant alternatives across varying experimental prompts. Finally, fine tuning establishes that a principal can explicitly engineer a target risk attitude.
\end{abstract}


\newpage

\section{Introduction}\label{sec:introduction}

Language models (i.e. LLMs) are increasingly delegated authority over decisions that allocate real resources. They are used to negotiate on the behalf of a principal \citep{araujo2026how}, to synthesize the information that determines the outcomes of lending \citep{eisfeldt2024ai}, and to advise households on saving and investing during the life cycle \citep{desilva2026ai}. Initially, language models gained attention for drafting text and meeting technical benchmarks. Now, especially with the introduction of AI agents, language models are resolving economic tradeoffs without direct input from a human principal. The economic question is no longer whether they are intelligent but whose preferences they impose. Importantly, a principal delegating an allocation problem to a language model cannot \emph{ex ante} specify her economic preferences for every unforeseen decision it will encounter. \cite{imas2025agentic} show experimentally that this gap is consequential. When many such delegated choices clear the same market, the risk attitude a model brings to an underspecified instruction stops being a private feature and shapes how resources are allocated, a transition \cite{shahidi2025coasean} argues is imminent as agents transact on behalf of consumers. That attitude is the object we measure. We study not a model reasoning toward a specified mandate but the default preference it reveals when the instruction leaves the choice over economic tradeoffs open, and we ask whether that default is stable enough to be treated as a preference at all.

We argue that a language model capable of reasoning need not be a coherent chooser. Typically, a model is judged as capable if it produces the known correct answer to a factual question. Choosing among competing economic ends is different. When those ends conflict, no choice is correct in the abstract because the tradeoff among them is settled by preference rather than by analysis, and what a principal delegates is precisely the authority to choose. A model can therefore reason well and still choose in ways that reveal no stable preference at all. This lack of stability becomes an explicit economic liability when choices violate, for example, the axiom of transitivity. If a language model's choices cycle according to a standard money pump, a counterparty can use it to systematically extract wealth from the principal. This is the distinction that matters for delegation. Therefore, our central question follows. What preferences govern a language model's choices over economic tradeoffs, and do those choices satisfy the revealed-preference restrictions required for a stable utility interpretation? Our answer, for the risk tradeoffs we study, is that they do, within limits we make precise. When the instruction leaves the tradeoff open, every model brings a systematic and measurable default risk attitude to the choice. However, we find the models' preferences are not invariant to how the options are presented to them.

We reach this answer with a single measurement program built on restricting how a language model chooses. Specifically, we require the model to express its choice by emitting a single unit of output, called a \emph{token}, as its revealed preference. The model selects that token through a fixed decoding rule that maps a vector of internal scores, its logits, into a probability over the available tokens. Our first contribution is to establish that this decoding rule admits an exact random utility representation, one in which the logits are the systematic utility index and an internal parameter called \emph{temperature} is the scale of an additive extreme-value shock. The representation is precisely the conditional logit of \citet{mcfadden1972conditional}, but its standing here is different. In the classical model, the systematic index is latent and must be inferred from choices under an assumed distribution for the unobserved shock. In a language model, the index is the logit vector itself, recorded directly from a single forward pass of the model, so the same quantity the discrete-choice econometrician estimates is here observed. The stochastic part of the representation imposed by temperature is only the decoding noise. Because the systematic index is observed, the identification of preferences is not based on the shock distribution that the classical estimator leans on.

Our experimental program that generated these insights has two steps. We first test whether a stable preference exists or whether the model's choices obey the revealed-preference restrictions required for a utility interpretation. Six diagnostics operationalize a set of standard revealed-preference axioms: completeness, reflexivity, monotonicity, transitivity, continuity, and independence of irrelevant alternatives. In the second step, we estimate preferences structurally in a controlled portfolio choice environment. In both steps, the model faces menus of assets, each described by an expected return and standard deviation, and selects one in a single forced choice, with the attributes varied exogenously across menus so that the return-risk tradeoff is entirely under the researcher's control. A maintained quadratic utility over the two attributes reduces the model's preference to a single risk-aversion parameter, which the estimator recovers. For \emph{open-weight} models, for which internal computation can be readily inspected, we read the utility index directly, so we observe the utility-relevant signal itself rather than reconstructing it from choices as in a human experiment \citep{ludwig2025large}. For proprietary \emph{frontier} models such as those from OpenAI, Anthropic and Google, whose internal computation is hidden, we recover the same quantities from repeated sampled choices using standard maximum likelihood, exactly as one would with human subjects. The replication code we provide is a protocol we recommend for testing agents along these dimensions of rationality.

Our findings across models demonstrate their choices respect the economic content of a menu but not its irrelevant features. Monotonicity, continuity, and transitivity hold at or near perfect rationality, so the models honor mean-risk dominance, respond smoothly to gradual changes in return and risk, and rank assets consistently. However, they fail to respond in a perfectly rational manner to the presentation of irrelevant alternatives, which is a form of changing the context of the decision. Moving an option's position in a menu can change the value a model attaches to it, and adding a strictly dominated third asset can shift the strength of preference between two unchanged options, thus reflexivity and independence of irrelevant alternatives are the weakest diagnostics throughout. The economic consequence of this error is that a model's preference ordering remains relatively stable while its intensity of preference does not, and only near indifference does this instability become observable. Furthermore, our structural estimations reveal that every model we evaluate is risk averse, and risk aversion varies across the major AI labs (OpenAI, Anthropic, Alibaba, etc.). Thus, the same menu implies materially different chosen portfolios depending on which lab's model is asked.

Although these findings broadly support treating language models as agents with stable and measurable preferences, they may not be optimal for a principal's delegated task. This raises the question of whether a desired economic preference can be installed in a model rather than merely measured in it. We show that it can. By fine-tuning, the standard procedure for adjusting a model's weights toward a specified objective,\footnote{This fine-tuning exercise differs from preference-based alignment pipelines that learn from human comparisons or optimize policies against learned preference models \citep{christiano2017deep,ouyang2022training,rafailov2023direct}. Related economic work fine-tunes language model agents toward explicit rational and moral preference structures in economic games and moral dilemmas \citep{lu2025aligning}.} we write the maintained utility into a specific open-sourced language model at preset risk-preference targets, and the model's induced choices match risk tolerance targets with high precision. Fine-tuning is therefore a direct lever for inducing a specified economic preference, expressed through the utility function, into the agent that will act on it.

Our contribution sits among several literatures. The first studies language models as economic agents and behavioral subjects. A growing body treats models as simulated decision-makers that asks how closely their behavior follows that of people \citep{horton2023large,mei2024turing,park2023generative, akata2025playing}, deploys them as synthetic respondents in applied research \citep{brand2023using}, and measures their psychological traits directly \citep{serapio2025psychometric}. Part of this work documents that model choices respond to surface features such as the order in which options are listed \citep{pezeshkpour2024large}, a sensitivity that our reflexivity and invariance diagnostics formalize and quantify. \cite{bini2026behavioral} study language model responses in preference-based tasks. Most of this literature describes behavior, whereas we estimate and test a preference. A second strand is the structural discrete choice and revealed-preference tradition on which our estimator rests \citep{afriat1967construction,mcfadden1972conditional,manski1977structure,varian1992microeconomic,berry1994estimating,train2009discrete,charness2013experimental,matvejka2015rational,demuynck2023computing}, which we apply to the directly observed internals of a generative model rather than to choices alone. A third is the experimental estimation of risk preferences from choices over alternatives with known payoff distributions \citep{holt2002risk,choi2007consistency,cohen2007estimating,harrison2008risk,andersen2008eliciting,bruhin2010risk,apesteguia2018monotone}, of which \citet{kim2024learning}, \citet{liu2025evaluating} and \citet{ouyang2024ethical} are the nearest applications to language models. \cite{ellis2026state} approach the same alignment problem from the principal's side. Our framing also connects to work on the delegation of decisions to algorithms and on when people accept or resist machine judgment \citep{dietvorst2015algorithm,logg2019algorithm,chugunova2022we}, to which we add a structural account of the preferences an algorithm brings once a decision is delegated to it.

The remainder of the paper proceeds as follows. Section~\ref{sec:structural-model} develops the random utility representation and the revealed-preference diagnostics. Section~\ref{sec:portfolio-choice} specializes the framework to the portfolio choice environment and derives the two estimators. Section~\ref{sec:experimental-design} describes the experimental design and implementation. Section~\ref{sec:results} reports the diagnostics, the structural estimates, and the fine-tuning validation. Section~\ref{sec:conclusion} concludes.

\section{Structural Model of Preference Formation}\label{sec:structural-model}

Suppose a researcher gave the following decision problem to a human subject.

\begin{tcolorbox}[
    colback=white,
    colframe=white,
    boxrule=0.6pt,
    sharp corners, 
    left=10pt,
    right=10pt,
    top=8pt,
    bottom=8pt,
    fontupper=\small\ttfamily 
  ]
  You have two options for your annual financial portfolio.
  Choose the portfolio you prefer and reply with A or B only.

  \vspace{0.5em}
  Option A offers an expected return of 5\% and a standard deviation\\
  of 10\%.

  \vspace{0.5em}
  Option B offers an expected return of 12\% and a standard deviation\\
  of 20\%.

  \vspace{1.2em}
  \hrule
  \vspace{1.2em}

  I choose Option \_\_\_\_\_
\end{tcolorbox}

From the perspective of classical choice theory, the exact mechanism driving the choice made by the subject is an unobserved latent process. For instance, it can be difficult to measure the strength of preference one has over the available options. Consequently, whatever the subject's underlying utility function is (if any), all aspects of preference must be identified through the observed decisions from the experiment. In contrast, because language models possess an explicit computational architecture that maps an option menu input to a choice, their decision-making apparatus is directly observable. In this sense, the preferences of a language model are easier to inspect than those of a human. This paper leverages this architectural transparency to identify not only the economic preferences revealed by a language model, but also specific components of the structural mechanism mapping information to choice. We begin by describing the relevant components of this architecture and how they relate to economic models of utility.

Modern transformer-based language models generate text autoregressively as a sequence of discrete tokens drawn from a finite vocabulary set $\mathcal{V}$ \citep{vaswani2017attention,brown2020language}.\footnote{Tokens are the fundamental units of text parsed by the model, encompassing words, subwords, and punctuation marks. The vocabulary set $\mathcal{V}$ is fixed by the model's architecture.} At each step in the sequence, the model evaluates the accumulated context of preceding tokens through a highly parameterized non-linear function comprising attention mechanisms and feed-forward networks. This evaluation yields a vector of latent unnormalized scores over the entire vocabulary. In the machine learning literature, these raw output values are referred to as ``logits.'' The realization of the subsequent token is then drawn stochastically from a categorical probability distribution over $\mathcal{V}$. The model derives this distribution by applying a softmax transformation to the vector of logits, mapping the unbounded latent scores into proper probabilities that sum to one.

In the discrete choice experiments that follow, we impose a one-step forced-choice problem in which a single output token from the model constitutes a revealed preference. This constrained decision problem invites comparison with the System 1 mode of cognition described by \citet{stanovich2000individual} and \citet{kahneman2003maps}. The resulting observed choices represent fast and associative heuristics rather than deliberative and sequential reasoning. Under this restricted choice architecture, we demonstrate that the model's preference rankings are completely determined by the logit differentials over the permissible action space.

For single-token revealed-preference problems, logit gaps provide the primitive object for identification. These gaps determine the relative odds of the menu labels under the decoding rule introduced below. With an affine utility specification, the same gaps become scaled utility differences. This maps the machine-learning object of a logit gap into the econometric object of a utility gap, placing our identification strategy in the structural choice tradition \citep{thurstone1927law,marschak1960binary,manski1975maximum}. We formalize the mapping by showing that the induced choice rule admits an exact random utility representation nested within the multinomial logit framework of \citet{mcfadden1972conditional}. We then define six revealed-preference indices that measure conformance to utility axioms and odds-invariance restrictions on the model's forced-choice outputs. The portfolio-choice problem in Section \ref{sec:portfolio-choice} derives the two estimators used for our structural recovery exercise.

\subsection{A Random Utility Representation}

For a given discrete choice problem, let $x$ denote the prompt context and let $C(x) \subseteq \mathcal{V}$ denote the set of label tokens (e.g. Option ``A'' vs ``B'')
comprising the feasible options
in the choice menu. Let $u_i(x)$ denote the raw logit on token $i \in \mathcal{V}$. Under standard softmax decoding \citep{bridle1989training}, the language model
induces a next-token probability distribution over the full vocabulary,
\begin{equation}\label{eq:softmax}
  P(i \mid x, \tau)
  = \frac{\exp(u_i(x) / \tau)}{\sum_{j \in \mathcal{V}} \exp(u_j(x) / \tau)}
  \qquad i \in \mathcal{V},
\end{equation}
where the parameter $\tau > 0$ is \emph{temperature}, which scales the dispersion of the induced choice distribution. As $\tau \downarrow 0$ the distribution concentrates all mass on the single token with the highest logit across $\mathcal{V}$. As $\tau \rightarrow \infty$ the distribution approaches the uniform over $\mathcal{V}$, spreading mass evenly across all vocabulary tokens and effectively destroying the model's inherent coherence.

\begin{proposition}\label{prop:rum}
  The softmax rule in Equation~\eqref{eq:softmax} admits an exact random utility representation,
  \[
    U_i(x,\tau) = u_i(x) + \tau \epsilon_i, \qquad i \in \mathcal{V},
  \]
  where the error terms $\epsilon_i$ are independent and identically distributed according to a standard Gumbel (Type I Extreme Value) distribution.
\end{proposition}

\begin{proof}
  See Appendix~\ref{app:proofs}.
\end{proof}

Since the shocks $\epsilon_i$ are i.i.d.\ Type I extreme value across all $i \in \mathcal{V}$, the softmax rule in Equation~\eqref{eq:softmax} is the full-vocabulary multinomial logit representation of the random utility model from \citep{mcfadden1972conditional}. Under this interpretation, a model's raw logit values serve as a systematic utility index, and the temperature parameter $\tau$ is the scale of the stochastic component. Consequently, as $\tau \downarrow 0$, choice converges to the label with the highest logit, and the distribution recovers only the ordinal ranking of the systematic utilities (i.e. the noise vanishes and the cardinal magnitudes of utility differences become irrelevant).

To connect logits to structural preferences, we specify the systematic utility index as a positive affine transformation of a cardinal utility index $V_i$ over the full vocabulary,\footnote{Under expected utility, any two cardinal representations of the same preferences differ by a positive affine transformation \citep{von1947theory}. Equation~\eqref{eq:logit-spec} therefore preserves the underlying vNM preference ordering and the associated risk attitudes.}
\begin{equation}\label{eq:logit-spec}
  u_i(x) = \kappa V_i(x) + \zeta(x), \qquad i \in \mathcal{V}.
\end{equation}
The scale parameter $\kappa > 0$ maps utility differences into logit units and is assumed constant across prompt contexts, so that scaled utility gaps remain comparable across the menus used in our experiments. The term $\zeta(x)$ is a prompt-specific common shifter that is the same for all tokens within a context and therefore cancels from every within-menu comparison. We leave the functional form of $V_i$ is unrestricted here and explore a quadratic form in Section~\ref{sec:portfolio-choice}. Nonetheless, its dependence on preference parameters already bears on identification, since \citet{apesteguia2018monotone} show that an additive random utility model of this form can fail to identify risk preferences when $V_i$ is a constant-risk-aversion expected utility whose between-prospect difference vanishes as risk aversion grows. We therefore require $V_i$ to be affine in the risk-preference parameter, which keeps that difference monotone in it.

\subsection{Random Utility Axioms for Language Models}

Having embedded the language model's token probabilities in a random utility framework, we now formulate the restrictions under which the token-level utility index can be interpreted as stable revealed preferences, and propose concomitant measures of restriction satisfaction. The transformer architecture supplies the logits over the vocabulary, and the random utility representation maps those logits into systematic utilities and induced choice probabilities. The axioms below are therefore not mechanical consequences of the architecture, but restrictions we impose and evaluate to determine whether the induced choices admit a stable revealed-preference interpretation. Throughout this subsection, for each token $k \in \mathcal{V}$, let
\[
  U_k(x,\tau) = \kappa V_k(x) + \zeta(x) + \tau \epsilon_k
\]
denote its realized random utility.

Following this utility specification, we define rational behavior according to the classic axioms of completeness, reflexivity, monotonicity, transitivty, and continuity presented in \cite{varian1992microeconomic}. Additionally, we define an axiom for the independence of irrelevant alternatives that is implied by the random utility model \citep{hausman1984specification}. Each definition is subsequently paired with a measure of adherence to the axiom that we employ in the experimental protocol that follows.

\begin{definition}[Completeness axiom]\label{def:completeness}
  Fix a prompt context $x$ and a choice menu $C(x) \subseteq \mathcal{V}$. The model satisfies \emph{menu completeness} on $C(x)$ at temperature $\tau>0$ if its realized maximizer lies in the intended menu with probability one. Equivalently,
  \[
    \Pr\!\left(
      \arg\max_{k \in \mathcal{V}} U_k(x,\tau) \in C(x)
    \right)=1
    .
  \]
\end{definition}

This means that any token chosen with positive probability is in the choice menu. For example, if the menu choices are A and B ($C(x)=\{A,B\}$) then perfect completeness would imply that the model puts no weight on any other tokens except for those representing ``A" and ``B". Traditionally, the completeness axiom dictates that the agent can provide a ranking of feasible options, allowing for indifference between bundles but not for answers that would fail to provide a preference ranking. Our condition operationalizes this requirement as menu adherence, in that the model must confine its response to the feasible options rather than produce a token that conveys no ranking of them.

\begin{measure}[Completeness index]\label{meas:completeness-index}
  The \emph{completeness index} measures the probability that the model's realized choice lies in the intended menu,
  \[
    \begin{aligned}
      \mathcal{C}(x,\tau)
      &\equiv
      \Pr\!\left(
        \arg\max_{k \in \mathcal{V}} U_k(x,\tau) \in C(x)
      \right) \\
      &=
      \sum_{i \in C(x)} P(i \mid x,\tau)
      =
      \frac{\sum_{i \in C(x)} \exp\!\left(\frac{\kappa}{\tau} V_i(x)\right)}{\sum_{k \in \mathcal{V}} \exp\!\left(\frac{\kappa}{\tau} V_k(x)\right)}.
    \end{aligned}
  \]
\end{measure}

The index $\mathcal{C}(x,\tau)$ lies in $[0,1]$ and equals one if and only if menu completeness holds on $C(x)$. Its dependence on temperature is only through the effective scale $\kappa/\tau$. Lower values of $\tau$ concentrate probability on tokens with maximal utility, while higher values of $\tau$ move the distribution toward the vocabulary share of $C(x)$. Because the full-vocabulary softmax places positive probability on every token whenever logits are finite, exact menu completeness is unattainable at any positive temperature. The axiom therefore describes an ideal, and the index measures proximity to that ideal. Exact completeness obtains only when the logits outside the menu are masked, or in the zero-temperature limit provided the token with the maximal logit lies in the intended menu. Given feasibility of the intended menu, the next restriction concerns equality. Reflexivity requires labels representing the same alternative to receive the same utility and hence the same choice probability.

\begin{definition}[Reflexivity axiom]\label{def:reflexivity}
  Fix a prompt context $x$ and a choice menu $C(x) \subseteq \mathcal{V}$, and let $\sim_x$ partition tokens in $C(x)$ that represent the same alternative. The model satisfies \emph{exact reflexivity} on $C(x)$ if, for all $i,j \in C(x)$,
  \[
    i \sim_x j
    \quad\Longrightarrow\quad
    V_i(x)=V_j(x)
    \quad\text{and hence}\quad
    P(i \mid x,\tau)=P(j \mid x,\tau).
  \]
\end{definition}

This condition operationalizes the reflexivity axiom as label invariance, in that the utility assigned to an alternative may not depend on the token label under which it is presented. If a language model is indifferent between option A and option B, and if the model is not biased to choose the first option,
then each label should receive one half of the choice probability, which we can observe directly in the logits.

The coordinate space in Figure \ref{fig:axiom_subplots} represents a choice space of portfolios that increase in risk from left to right and have higher expected return moving up the horizontal axis. Thus, relative to the reference portfolio (gold star), points directly above it should be preferable to points below. The color gradient shows the progression of logit differences. Dark blue represents an area where the portfolios are intensely preferred to the reference star because those points offer both a much higher return and lower risk. The solid black line maps an indifference curve, or all the points that are valued equally as the reference point. A risk averse decision maker would exhibit an upward sloping indifference curve in this choice space, as shown in the picture. Panel (a) of Figure \ref{fig:axiom_subplots} shows that if A and B label the same portfolio, then they should receive the same utility regardless of which label is assigned.

\begin{figure}[htbp]
  \centering
  \includegraphics[width=\textwidth,height=0.86\textheight,keepaspectratio]{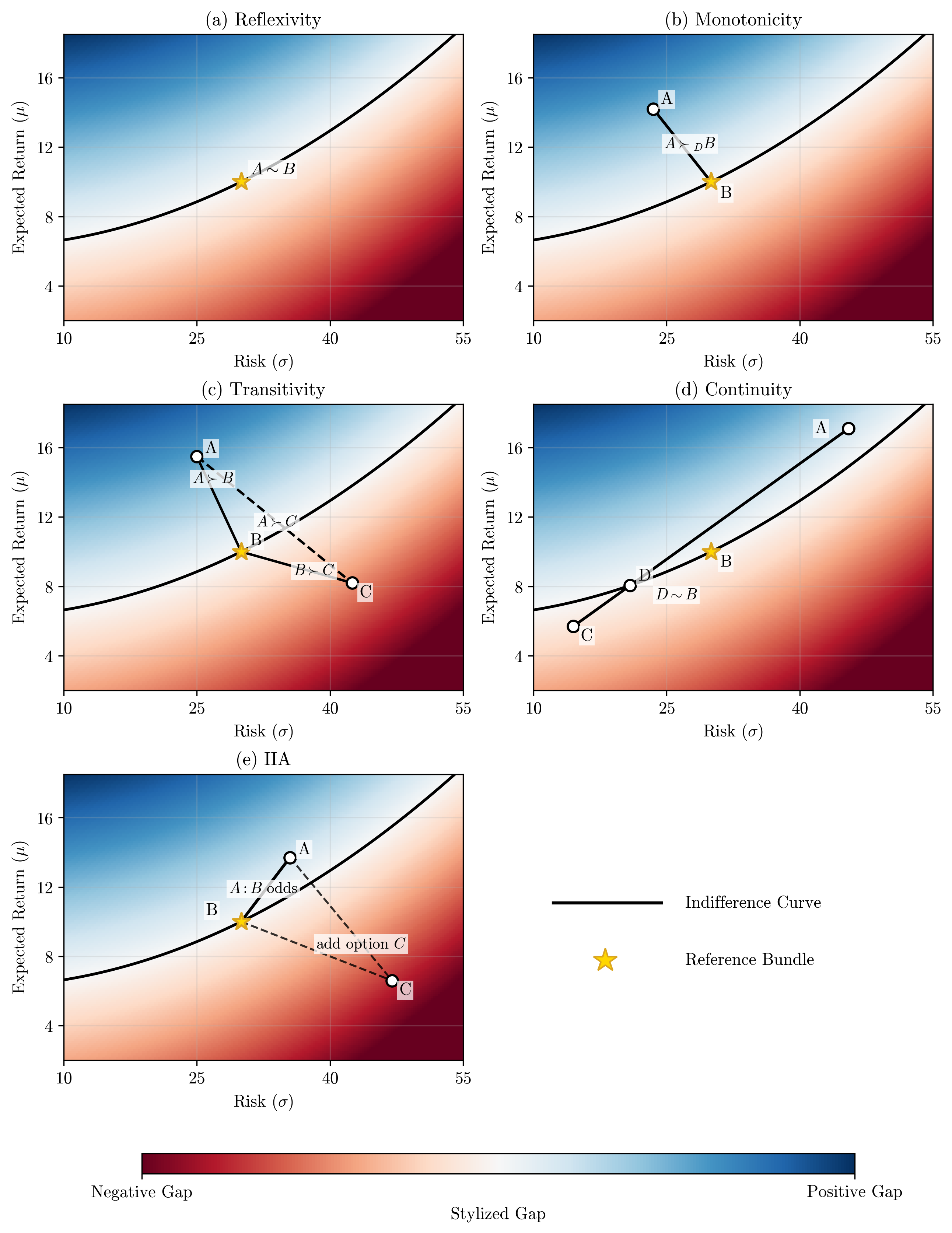}
  \caption{Stylized visual representations of the random utility restrictions evaluated in the single-token forced-choice portfolio problems.}
  \label{fig:axiom_subplots}
\end{figure}

\begin{measure}[Reflexivity index]\label{meas:reflexivity-index}
  Fix $i,j \in C(x)$ such that $i\sim_x j$. The \emph{reflexivity index} for pair $(i,j)$ is
  \[
    \begin{aligned}
      \mathcal{R}(i,j,x,\tau)
      &=
      1
      -
      \frac{
        \left|P(i \mid x,\tau)-P(j \mid x,\tau)\right|
      }{
        P(i \mid x,\tau)+P(j \mid x,\tau)
      } \\
      &=
      \frac{2}{
        1+
        \exp\!\left(
          \frac{\kappa}{\tau}
          \left|V_i(x)-V_j(x)\right|
        \right)
      },
    \end{aligned}
  \]
  with the second equality following from the odds ratio $\frac{P(i \mid x,\tau)}{P(j \mid x,\tau)}=\exp\!\left(\frac{\kappa}{\tau}\bigl[V_i(x)-V_j(x)\bigr]\right)$.
\end{measure}

The index $\mathcal{R}(i,j,x,\tau)$ lies in $[0,1]$ and equals one if and only if $V_i(x)=V_j(x)$, equivalently $P(i \mid x,\tau)=P(j \mid x,\tau)$. Temperature enters only through the scaled utility difference. If $V_i(x)=V_j(x)$, exact reflexivity is invariant to temperature. If $V_i(x)\neq V_j(x)$, lowering $\tau$ increases the scaled gap and drives $\mathcal{R}(i,j,x,\tau)$ toward zero, while raising $\tau$ compresses the gap and moves $\mathcal{R}(i,j,x,\tau)$ toward one.

The next restriction concerns strict improvement. Given an exogenous mean-risk dominance order, monotonicity requires the model's induced order to respect every dominance comparison. Panel (b) of Figure \ref{fig:axiom_subplots} shows that portfolio A has both a higher expected return, $\mu$, and less uncertainty, relative to porfolio B. Therefore, any agent whose preferences are increasing in expected return and decreasing in risk should prefer A to B. (Rationality alone does not deliver this ranking, since a risk-seeking expected utility maximizer may prefer the riskier portfolio.) We use $\succ_D$ to denote this mean-risk dominance relation. The benchmark utility index maintained in Section~\ref{sec:portfolio-choice} satisfies this monotonicity requirement for every positive value of its risk preference parameter on the mean-risk dominance menus used in our design, where the dominating portfolio has higher expected return and lower raw second moment $R$.

\begin{definition}[Monotonicity axiom]\label{def:monotonicity}
  Fix a prompt context $x$ and a menu $C(x) \subseteq \mathcal{V}$, and let $\succ_D$ be an exogenous mean-risk dominance order on $C(x)$. The model satisfies \emph{exact monotonicity} on $C(x)$ if, for all $i,j \in C(x)$,
  \[
    i \succ_D j
    \quad\Longrightarrow\quad
    V_i(x)>V_j(x)
    \quad\text{and hence}\quad
    P(i \mid x,\tau)>P(j \mid x,\tau).
  \]
\end{definition}

\begin{measure}[Monotonicity index]\label{meas:monotonicity-index}
  Fix $i,j \in C(x)$ such that $i \succ_D j$. The \emph{monotonicity index} for pair $(i,j)$ is
  \[
    \begin{aligned}
      \mathcal{M}(i,j,x,\tau)
      &=
      \mathbf{1}\!\bigl[P(i \mid x,\tau)>P(j \mid x,\tau)\bigr] \\
      &=
      \mathbf{1}\!\bigl[V_i(x)>V_j(x)\bigr].
    \end{aligned}
  \]
  Then $\mathcal{M}(i,j,x,\tau) \in \{0,1\}$, with $\mathcal{M}(i,j,x,\tau)=1$ if and only if the model satisfies monotonicity for the dominance pair $(i,j)$.
\end{measure}

Our measure requires that the probability of choosing a preferable portfolio ($i$) be higher than the probability of choosing a less desirable portfolio ($j$). A language model can place positive probability on the mean-risk dominated option without violating the monotonicity axiom because the axiom restricts the ordering of systematic utilities and choice probabilities, not the probability level assigned to each option. For example, considering panel (b) of Figure \ref{fig:axiom_subplots}, the probability assigned to portfolio B can be positive as long as the probability assigned to A is higher. Monotonicity therefore is a pairwise restriction relative to the exogenous order $\succ_D$. It requires the model's induced order to agree with $\succ_D$ on each dominance pair, but it imposes no condition on the joint consistency of distinct pairwise comparisons. The next axiom imposes that consistency condition by requiring the strict relation induced by the model to be closed under composition (represented in panel (c) of Figure \ref{fig:axiom_subplots}).

\begin{definition}[Transitivity axiom]\label{def:transitivity}
  Fix a collection of alternatives $\{a,b,c\}$ and, for each pair $(i,j)$, let $x_{ij}$ denote the prompt context presenting the binary menu over $i$ and $j$. Let $\succ$ be the strict relation induced by the model across these pairwise contexts, where
  \begin{equation}\label{eq:transitivity-relation}
    i \succ j
    \quad\Longleftrightarrow\quad
    V_i(x_{ij})>V_j(x_{ij})
    \quad\Longleftrightarrow\quad
    P(i \mid x_{ij},\tau)>P(j \mid x_{ij},\tau).
  \end{equation}
  The model satisfies \emph{exact transitivity} on the triad $\{a,b,c\}$ if
  \[
    a \succ b
    \quad\text{and}\quad
    b \succ c
    \quad\Longrightarrow\quad
    a \succ c.
  \]
\end{definition}

Within a single prompt context the index $V_i(x)$ assigns a real number to every token, so the ordering induced inside one menu is transitive by construction. The empirical content of the axiom therefore concerns whether the rankings elicited from separate binary menus can be rationalized by a single stable ordering over the underlying alternatives.

\begin{measure}[Transitivity index]\label{meas:transitivity-index}
  For a closed triad $t=\{a,b,c\}$, define
  \[
    \Delta^{V}_{ij}=V_i(x_{ij})-V_j(x_{ij}),
    \qquad
    \Delta^{P}_{ij}=P(i \mid x_{ij},\tau)-P(j \mid x_{ij},\tau)
  \]
  for each pairwise leg $(i,j) \in \{(a,b),(b,c),(a,c)\}$. Since $u_i(x)=\kappa V_i(x)+\zeta(x)$ with $\kappa>0$, the sign of $\Delta^{V}_{ij}$ equals the sign of $\Delta^{P}_{ij}$ and induces the pairwise ordering of $i$ and $j$. The triad is \emph{directional} if the first two legs imply a strict chain,
  \[
    \Delta^{V}_{ab} > 0 \ \wedge\  \Delta^{V}_{bc} > 0,
    \qquad \text{or} \qquad
    \Delta^{V}_{ab} < 0 \ \wedge\  \Delta^{V}_{bc} < 0.
  \]
  A directional triad is \emph{consistent} if the closure leg confirms the implied ranking, namely $\Delta^{V}_{ac}$ has the same sign as $\Delta^{V}_{ab}$ and $\Delta^{V}_{bc}$. For a directional triad $t$, the \emph{transitivity index} is
  \[
    \begin{aligned}
      \mathcal{T}(t)
      &\equiv
      \mathbf{1}\!\left[
        \Delta^{P}_{ab}\Delta^{P}_{bc}>0
        \ \wedge\
        \Delta^{P}_{ab}\Delta^{P}_{ac}>0
      \right] \\
      &=
      \mathbf{1}\!\left[
        \Delta^{V}_{ab}\Delta^{V}_{bc}>0
        \ \wedge\
        \Delta^{V}_{ab}\Delta^{V}_{ac}>0
      \right].
    \end{aligned}
  \]
  Then $\mathcal{T}(t) \in \{0,1\}$, with $\mathcal{T}(t)=1$ if and only if the closure leg confirms the implied ranking. A non-directional triad places no restriction on the closure leg, so $\mathcal{T}(t)$ is scored only for directional triads.
\end{measure}

Panel c of Figure \ref{fig:axiom_subplots}) shows that if A lies above the indifference curve for bundle B (the gold star) and C lies below that indifference curve, then it follows that this decision maker would prefer A to C, if they have rational preferences. Therefore, transitivity is a finite consistency requirement on a closed triad. If the first two comparisons induce a strict chain, then the third comparison must agree with the order obtained by composing the first two.

The next axiom is continuity, which rules out discontinuous jumps in preference. Standard preference continuity implies that, between alternatives ranked on opposite sides of a reference alternative, a continuous path must cross the reference alternative's indifference contour.

\begin{definition}[Continuity axiom]\label{def:continuity}
Fix alternatives $a,b,c$ with $a \succ b \succ c$ under the pairwise relation in Equation~\eqref{eq:transitivity-relation}, and let $p$ denote the continuous interpolation path
\[
  \gamma(t)=(1-t)c+ta,
  \qquad t \in [0,1],
\]
connecting $c$ to $a$ in attribute space. For each $t$, let $x(t)$ denote the prompt context presenting the binary menu over $\gamma(t)$ and the reference alternative $b$. Let $\ell_{\gamma(t)}(x(t))$ and $\ell_b(x(t))$ denote the label tokens in that context that encode those two bundles. The model satisfies \emph{exact continuity} on $p$ if there exists $t^\ast \in (0,1)$ such that
\[
  V_{\ell_{\gamma(t^\ast)}}(x(t^\ast))
  =
  V_{\ell_b}(x(t^\ast)),
\]
and hence
\[
  P(\ell_{\gamma(t^\ast)} \mid x(t^\ast),\tau)
  =
  P(\ell_b \mid x(t^\ast),\tau),
\]
so that $\gamma(t^\ast) \sim b$.
\end{definition}

\begin{measure}[Continuity index]\label{meas:continuity-index}
For a path $p$, define the log odds gap
\[
  g^{P}(t)=
  \log
  \frac{P(\ell_{\gamma(t)} \mid x(t),\tau)}{P(\ell_b \mid x(t),\tau)}.
\]
The population \emph{continuity index} is
\[
  \mathcal{K}(p)
  \equiv
  \mathbf{1}\!\left[
    \exists t^\ast \in (0,1)
    \text{ such that }
    g^{P}(t^\ast)=0
  \right].
\]
Then $\mathcal{K}(p) \in \{0,1\}$, with $\mathcal{K}(p)=1$ if and only if path $p$ intersects the indifference contour of $b$. The population index conditions on the entire interval $[0,1]$ and is therefore not observable from finitely many probed comparisons. The empirical implementation evaluates a grid version of the index. For probed points $0 < t_1 < \cdots < t_m < 1$, define
\[
  \mathcal{K}_G(p)
  \equiv
  \mathbf{1}\!\left[
    \min_{r=1,\ldots,m} g^{P}(t_r)
    \leq 0
    \leq
    \max_{r=1,\ldots,m} g^{P}(t_r)
  \right],
\]
which equals one if and only if the probed log odds bracket zero, providing a finite-grid analogue of the continuity implication.
\end{measure}

Continuity is the corresponding path-existence restriction. It requires an indifference crossing of the reference alternative along the interpolation path, equivalently a zero crossing of the log probability odds. The grid index records whether the probed comparisons bracket such a crossing, without imposing monotonic movement or uniqueness.

The final restriction concerns the stability of pairwise odds when alternatives outside the pair are varied. Within a fixed prompt context, odds invariance is a direct consequence of the model's equivalence to the standard random utility model, where the ratio of choice probabilities for any two options depends solely on their systematic utility difference \citep{luce1959individual}. Across the binary and expanded prompts used below, the restriction becomes an empirical test of whether that systematic utility difference remains stable when the menu changes.

\begin{definition}[Independence of irrelevant alternatives]\label{def:iia}
Fix a prompt context $x$ with $a,b \in C(x)$ and an expanded context $x^c$ with $a,b,c \in C(x^c)$, where $x^c$ differs from $x$ only by adding a strictly mean-risk dominated alternative $c$ such that $a \succ_D c$ and $b \succ_D c$. The model satisfies \emph{exact independence of irrelevant alternatives} for pair $(a,b)$ if
\[
  V_a(x)-V_b(x)
  =
  V_a(x^c)-V_b(x^c)
  \quad\Longleftrightarrow\quad
  \frac{P(a \mid x,\tau)}{P(b \mid x,\tau)}
  =
  \frac{P(a \mid x^c,\tau)}{P(b \mid x^c,\tau)}.
\]
\end{definition}

\begin{measure}[Independence of irrelevant alternatives index]\label{meas:iia-index}
For a mean-risk dominated alternative IIA comparison $q=(a,b,c,x,x^c)$ with $a \succ_D c$ and $b \succ_D c$, define
\[
  \begin{aligned}
    D^{V}(q)
    &=
    \bigl[V_a(x)-V_b(x)\bigr]
    -
    \bigl[V_a(x^c)-V_b(x^c)\bigr], \\
    D^{P}(q,\tau)
    &=
    \log
    \frac{P(a \mid x,\tau)}{P(b \mid x,\tau)}
    -
    \log
    \frac{P(a \mid x^c,\tau)}{P(b \mid x^c,\tau)}.
  \end{aligned}
\]
Since $D^{P}(q,\tau)=(\kappa/\tau)D^{V}(q)$, the \emph{independence of irrelevant alternatives index} for comparison $q$ is
\[
  \begin{aligned}
    \mathcal{I}(q,\tau)
    &\equiv
    \frac{2}{
      1+
      \exp\!\left(
        \left|D^{P}(q,\tau)\right|
      \right)
    } \\
    &=
    \frac{2}{
      1+
      \exp\!\left(
        \frac{\kappa}{\tau}
        \left|D^{V}(q)\right|
      \right)
    }.
  \end{aligned}
\]
Then $\mathcal{I}(q,\tau) \in [0,1]$, with $\mathcal{I}(q,\tau)=1$ if and only if the pairwise log odds of $a$ against $b$ are invariant to the mean-risk dominated alternative $c$.
\end{measure}

Independence of irrelevant alternatives is an odds-invariance restriction for an irrelevant alternative introduced by construction. It requires the pairwise log odds of $a$ against $b$ to be unchanged when the menu is expanded by an alternative $c$ that is strictly mean-risk dominated by both $a$ and $b$.

Together, the six indices state the restrictions under which the random utility representation and the affine logit specification admit a revealed-preference interpretation in one-step token choice. We next apply this framework to the portfolio choice problem and derive the two structural estimators.

\section{Portfolio Choice Problem}\label{sec:portfolio-choice}

The structural framework in Section~\ref{sec:structural-model} connects the model's logit vector to a cardinal utility index $V_i$ through the scale parameter $\kappa$, but leaves the functional form of $V_i$ unrestricted. We now specialize that index to a portfolio choice environment in which each alternative is a one-period asset with observed expected return $\mu_i$ and standard deviation $\sigma_i$, so that the relevant tradeoff between return and risk is transparent, fully observed by the researcher, and directly presented to the model. The researcher then recovers variance and the raw second moment mechanically from the reported moments. This environment gives the structural parameters an economically interpretable domain while keeping the objects of choice under complete experimental control.

In this environment, each prompt context $x_j$ presents a binary feasible portfolio set $C(x_j) = \{A,B\}$, where the two label tokens encode portfolios with observed central moments $(\mu_{Aj}, \sigma_{Aj})$ and $(\mu_{Bj}, \sigma_{Bj})$. Holding the rest of the prompt template fixed while varying these moments across menus gives a direct implementation of the measurement objects from Section~\ref{sec:structural-model} in a setting with economically interpretable alternatives. Consequently, the associated axioms reveal whether language models exhibit preferences consistent with a single utility function. Relatedly, we note the risk preferences recovered in this setting reflect a preference index over a well-defined class of risky allocations, but one that need not summarize behavior outside it. It is well-known that risk preferences are not globally consistent across domains among human subjects \citep{einav2012general, harrison2007naturally, schildberg2018risk}, and we explore the extent to which this is true for language models as well.

\subsection{Benchmark Utility Representation}

Our maintained benchmark is quadratic expected utility, which serves as a useful approximation to the expected utilities of a diverse set of utility functions \citep{markowitz2014mean}. To express this compactly, let $R_i = E[x_i^2]$ denote the raw second moment of portfolio returns. The moments shown in the prompt then imply $R_i = \sigma_i^2 + \mu_i^2$, since $\sigma_i^2 = E[x_i^2] - E[x_i]^2$. Recovering structural risk preference parameters from discrete choices over financial alternatives with known payoff distributions has clear precedent in the empirical literature \citep{mester1996study, choi2007consistency, cohen2007estimating, list2011ceos}.\footnote{Richer alternatives such as prospect theory, probability weighting, and disappointment aversion have been studied extensively \citep{kahneman1979prospect, gul1991theory, prelec1998probability, bruhin2010risk}.} For the label tokens that encode portfolios, we model the corresponding expected utility index as
\begin{equation}\label{eq:eu}
  V_i = \mu_i - \beta R_i,
\end{equation}
where $\beta > 0$ indicates risk aversion, $\beta < 0$ indicates risk seeking, and $\beta = 0$ collapses to expected-return maximization. Under Equation~\eqref{eq:eu}, whether portfolio $i$ is preferred to portfolio $i'$ when $\mu_i > \mu_{i'}$ and $R_i > R_{i'}$ depends entirely on whether $\beta \leq (\mu_i - \mu_{i'})/(R_i - R_{i'})$, so differences in $\beta$ translate directly into different portfolio rankings and, once the model selects among feasible allocations, different economic outcomes.

Substituting Equation~\eqref{eq:eu} into Equation~\eqref{eq:logit-spec} for the portfolio labels maps the observed portfolio moments into the pre-temperature logit index. Specifically, for alternative $i$ in menu $j$,
\begin{equation}\label{eq:logit-moments}
  u_{ij} = \zeta_j + \kappa \mu_{ij} - \kappa \beta R_{ij},
\end{equation}
where the additive term $\zeta_j$ is common within menu $j$ and drops from within-menu comparisons. By imposing strategic variation in $(\mu_{ij}, R_{ij})$ across alternatives and menus, we can identify the structural parameters of interest. We therefore consider two observation regimes using this framework. When raw pre-softmax logits on the menu labels are observed, within-menu logit gaps can be used to identify the structural parameters directly. When only sampled choices are observed, we employ traditional random utility model estimation methods via maximum likelihood. We explore each of these in turn below. Appendix~\ref{app:experimental-design} summarizes the identification conditions and design implications under both observation regimes.

\subsection{Observed-Logit Estimation}

When logits are observed, the data for menu $j$ are the two label logits together with the corresponding portfolio moments. Because the additive menu-specific level in Equation~\eqref{eq:logit-moments} is common to both options, estimation proceeds on the within-menu logit gap and does not require menu fixed effects. An equivalent stacked representation is
\[
  u_{ij} = \theta_{0j} + \alpha \mathbf{1}\{i=A\} + \theta_1 \mu_{ij} + \theta_2 R_{ij} + \varepsilon_{ij}
\]
where $u_{ij}$ is the raw logit for alternative $i$ in menu $j$, $\theta_{0j}$ absorbs the menu-specific level including the common term $\zeta_j$ from Equation~\eqref{eq:logit-moments}, and $\alpha$ captures any systematic additive position bias of the first option labeled A. Position bias is a well-documented phenomenon of current language models, which we discuss in Section \ref{sec:experimental-design} in further detail. With two observations per menu, this stacked equation is algebraically equivalent to the first-difference regression used for estimation.

For a binary menu $j$ with options A and B, within-menu differencing yields
\begin{equation}\label{eq:first-difference}
  \Delta u_j = \alpha + \theta_1 \Delta \mu_j + \theta_2 \Delta R_j + \eta_j
\end{equation}
where $\Delta u_j = u_{Aj} - u_{Bj}$, $\Delta \mu_j = \mu_{Aj} - \mu_{Bj}$, $\Delta R_j = R_{Aj} - R_{Bj}$, and $\eta_j = \varepsilon_{Aj} - \varepsilon_{Bj}$. Estimation therefore runs on the logit gap with a pooled constant and the within-menu differences in moments. Because the dependent variable is the raw pre-temperature logit gap, decoding temperature does not affect the identifying equation or the precision of this estimator. In this sense, the estimator keeps the component of the logit gap that is systematically explained by portfolio attributes and treats the remaining gap as either positional bias, absorbed by $\alpha$ or removed by mirroring, or specification error.

Identification requires that the slope mapping from utility units to logit units be common across menus, that any positional effect enter as a common additive constant, and that the design matrix with rows $(\Delta \mu_j,\Delta R_j)$ have full rank. The reduced-form slopes imply
\[
  \hat{\kappa} = \hat{\theta}_1
\]
and
\[
  \hat{\beta} = -\frac{\hat{\theta}_2}{\hat{\theta}_1}.
\]
Inference can use heteroskedasticity-robust standard errors for $(\hat{\theta}_1,\hat{\theta}_2)$. Since $\hat{\beta}$ is a ratio estimator, inference for $\hat{\beta}$ follows from the Delta method applied to the estimated covariance matrix.

\subsection{Sampled-Choice Estimation}

When only sampled choices are observed, we restrict attention to valid A/B responses and let $y_{jr}=1$ if option A is chosen on valid draw $r$ from menu $j$ and $y_{jr}=0$ if option B is chosen. With $\Delta \mu_j = \mu_{Aj} - \mu_{Bj}$ and $\Delta R_j = R_{Aj} - R_{Bj}$, the binary choice probability is
\begin{equation}\label{eq:revealed-prob}
  p_j \equiv P(y_{jr} = 1 \mid \Delta \mu_j, \Delta R_j)
  = \Lambda(\gamma_1 \Delta \mu_j + \gamma_2 \Delta R_j),
\end{equation}
where $\Lambda(t) = \exp(t) / (1 + \exp(t))$, $\gamma_1 = \kappa / \tau$, and $\gamma_2 = -\kappa\beta / \tau$. Note that varying $\tau$ across menus does not add identifying content because it only rescales the common latent index. Therefore, in the benchmark design we hold $\tau = 1$ fixed across estimation menus, which preserves the native logit scale.

Full-rank variation in $(\Delta \mu_j,\Delta R_j)$ identifies the reduced-form slopes. Choice data therefore identify $\beta = -\gamma_2/\gamma_1$ whenever $\gamma_1 \neq 0$. In the frontier regime (i.e. ChatGPT from OpenAI) of our experiments, let $n_j$ denote the number of valid A/B responses from menu $j$. Let $Y_j=\sum_{r=1}^{n_j} y_{jr}$ denote the number of those responses choosing option A. Estimation maximizes the grouped binomial log-likelihood
\begin{equation}\label{eq:log-likelihood}
  \mathcal{L}(\gamma_1, \gamma_2)
  =
  \sum_{j=1}^J
  \left[
    Y_j \ln p_j
    +
    (n_j-Y_j)\ln(1-p_j)
  \right].
\end{equation}
The omitted binomial coefficient is constant in $(\gamma_1,\gamma_2)$ and therefore does not affect the maximizer. Under these conditions, together with the standard identification and regularity conditions stated above, maximum likelihood yields consistent estimates $\hat{\gamma}_1$ and $\hat{\gamma}_2$ \citep{newey1994large}. Because $\tau$ is exogenously imposed, the structural parameters can be recovered as
\[
  \hat{\kappa} = \tau \hat{\gamma}_1
\]
and
\[
  \hat{\beta} = -\frac{\hat{\gamma}_2}{\hat{\gamma}_1}.
\]
Inference for $\hat{\kappa}$ follows from the inverse Hessian, while for $\hat{\beta}$ the Delta method once again is convenient away from weak-ratio cases.

\section{Experimental Design}\label{sec:experimental-design}

The experiment is built to address two measurement problems. The first is to determine whether a model's single-token choices satisfy the revealed-preference restrictions required for the utility representation defined by our measures of completeness, reflexivity, monotonicity, transitivity, continuity, and independence-of-irrelevant-alternatives formalized in Section~\ref{sec:structural-model}. The second is to estimate the shape of preferences under the quadratic mean-risk benchmark of Section~\ref{sec:portfolio-choice}, recovering the scale parameter $\kappa$ and the risk preference parameter $\beta$.

The unit of observation throughout is a model-menu prompt, a single decision context $x$ that presents a labeled portfolio menu and to which the model responds by placing probability on one label token. What the researcher observes from that prompt differs across the two regimes studied in the paper. For open-weight models we record the raw pre-softmax logits on the label tokens from a single forward pass.\footnote{All open-weight models are accessed through their publicly released weights. The full raw logit vector is recorded at the label token positions before any softmax decoding or sampling rule is applied. No provider-side post-processing intervenes between the final layer output and the recorded logit vector, so the observed-logit regime is exact.} For private frontier models we typically do not observe logits and instead query the model repeatedly under softmax decoding, recording the sampled choices from which empirical choice probabilities are formed.

\subsection{Model Subjects and Elicitation Protocol}

Our subject pool spans representative samples of non-thinking models from the open- and closed-source regimes. Within each regime we include models that were publicly available during the window of our study in June 2026, exposed stable version identifiers, and could be queried under a one-token forced-choice protocol with fixed decoding controls. Table~\ref{tab:model-subjects} lists the resulting sample. Panel~A covers open-weight models, and Panel~B covers private frontier models. Full details on the experimental code, hardware, vendor metadata, prompt and design hashes, and run identifiers are collected in the replication materials in Appendix~\ref{app:experimental-design}.

\begin{table}[ht]
  \centering
  \small
  \resizebox{\textwidth}{!}{
    \begin{tabular}{lllc}
      \toprule
      Model Identifier & Developer & API / Hugging Face Identifier & Size (Params) \\
      \midrule
      \multicolumn{4}{l}{\textit{Panel A. Open-Weight Models}} \\
      \midrule
      Gemma-2-9B & Google & \texttt{google/gemma-2-9b-it} & 9B \\
      Gemma-4-31B & Google & \texttt{google/gemma-4-31B-it} & 31B \\
      Llama-3.1-8B & Meta & \texttt{meta-llama/Llama-3.1-8B-Instruct} & 8B \\
      Llama-3.1-70B & Meta & \texttt{meta-llama/Llama-3.1-70B-Instruct} & 70B \\
      Ministral-3-8B & Mistral AI & \texttt{mistralai/Ministral-3-8B-Instruct-2512} & 8B \\
      Qwen3-8B & Alibaba & \texttt{Qwen/Qwen3-8B} & 8B \\
      Qwen3-14B & Alibaba & \texttt{Qwen/Qwen3-14B} & 14B \\
      \midrule
      \multicolumn{4}{l}{\textit{Panel B. Private Frontier Models}} \\
      \midrule
      GPT-4o & OpenAI & \texttt{gpt-4o-2024-08-06} &  \\
      GPT-4o mini & OpenAI & \texttt{gpt-4o-mini-2024-07-18} &  \\
      Claude 4.5 Haiku & Anthropic & \texttt{claude-haiku-4-5-20251001} &  \\
      Claude 4.6 Sonnet & Anthropic & \texttt{claude-sonnet-4-6} &  \\
      Gemini 2.5 Flash Lite & Google & \texttt{gemini-flash-lite-latest} &  \\
      \bottomrule
    \end{tabular}
  }
  \caption{Model subject pool under evaluation.}
  \label{tab:model-subjects}
\end{table}

The main experiment administers a common menu set to every model. The decision prompt follows a fixed three-part template as presented in Figure~\ref{fig-prompt-template} below. The first part frames the decision as an annual financial portfolio choice and instructs the model to reply with its preferred label from the listed options. The second part presents the options by specifying the expected return and standard deviation of each. The third part is the cue suffix ``I choose Option''. For private frontier models this suffix is appended to the end of the user prompt. For open-weight models it is prefilled at the start of the assistant response block. In both cases the model is induced to emit a single token indicating its choice. In all prompts the labels are the capital letters A and B (and C for the trinary independence menu), which are single tokens in the vocabulary of every model in the subject pool.\footnote{Some tokenizers include leading-space variants of the A and B tokens. When more than one token realization maps to a label, we aggregate at the label level by summing the softmax probabilities of the associated token realizations. The equivalent grouped logit at temperature $\tau$ is $\tau\log\sum_r\exp(u_r/\tau)$ over realizations $r$ of the label. We verify this property before running the experiment and report the token IDs in Appendix~\ref{app:experimental-design}.}\footnote{We explored the potential effects of using these specific option labels by randomly selecting tokens as option labels. We found no meaningful difference in these experimental results.}

\begin{figure}[ht]
  \centering
  \begin{tcolorbox}[
      colback=white,
      colframe=black,
      boxrule=0.6pt,
      sharp corners,
      left=10pt,
      right=10pt,
      top=8pt,
      bottom=8pt,
      fontupper=\small\ttfamily
    ]
    You have two options for your annual financial portfolio.
    Choose the portfolio you prefer and reply with A or B only.

    \vspace{0.5em}
    Option A offers an expected return of [$\mu_A$]\% and a standard deviation\\
    of [$\sigma_A$]\%.

    \vspace{0.5em}
    Option B offers an expected return of [$\mu_B$]\% and a standard deviation\\
    of [$\sigma_B$]\%.

    \vspace{1.2em}
    \hrule
    \vspace{1.2em}

    I choose Option
  \end{tcolorbox}
  \caption{Decision prompt template}
  \label{fig-prompt-template}
\end{figure}

The two observation regimes differ not only in what is recorded but in how choice probabilities are formed. In the open-weight regime, one forward pass per menu yields the exact label logits, and no replication is required. In the frontier regime, choice probabilities are formed either from repeated sampled choices or, where the provider exposes them, from first-token logprobs. Across both regimes we fix the decoding controls at temperature $\tau = 1$ and top-$p = 1$.\footnote{Top-$p$ sampling (also known as nucleus sampling) restricts the candidate token set at each step of decoding to the smallest subset of tokens whose cumulative probability exceeds the threshold $p$ \citep{holtzman2019curious}. Setting top-$p=1$ keeps the entire vocabulary and avoids truncation, which would break the full-vocabulary softmax that underlies the random utility representation in Proposition~\ref{prop:rum} and the indices that rest on it \citep{mcfadden1972conditional, train2009discrete}.}

\subsection{Portfolio Menu Construction and Diagnostic Design}\label{sec:exp-menus}

All menus draw choice options from a discrete grid of portfolio bundles with expected returns $\mu \in \{2, 3, \ldots, 18\}$ percent and standard deviations $\sigma \in \{10, 15, 20, \ldots, 55\}$ percent, yielding $17 \times 10 = 170$ distinct bundles, each with raw second moment $R = \sigma^2 + \mu^2$. This range spans low-return, low-risk profiles through high-return, high-risk profiles, and the spacing is fine enough to generate meaningful variation in the risk-return tradeoff. The $170$ bundles admit $\binom{170}{2} = 14{,}365$ distinct binary pairings, of which $13{,}596$ form the admissible working pool after removing pairs that are degenerate for the relevant first differences. From this pool we strategically construct the diagnostic and estimation menus described below, holding the menu sets constant across all models.

To control for a well known position bias exhibited by many language models \citep{zheng2024large,pezeshkpour2024large,guan2025order}, and for the broader possibility that choices depend on the frame in which a fixed alternative is presented \citep{salant2008f}, every canonical menu is also administered in a mirrored presentation that swaps which option carries each label. The reflexivity placebos described below show that a label's position can carry a systematic logit premium orthogonal to the systematic utility index in Equation~\eqref{eq:logit-spec}. For each menu we therefore form the position-corrected log-odds gap
\begin{equation}\label{eq:position-correction}
  \Delta z = \tfrac{1}{2}\bigl[\Delta u - \Delta u^{\text{mirror}}\bigr],
\end{equation}
where $\Delta u$ and $\Delta u^{\text{mirror}}$ denote the A-minus-B log-odds gap under the canonical and mirrored presentations. Because the positional premium enters both presentations with the same sign while the utility difference reverses, the half difference removes the positional component and retains the systematic component. Unless stated otherwise, the monotonicity, transitivity, continuity, and independence diagnostics are evaluated on these position-corrected gaps. The reflexivity index is computed within a single presentation by construction, since the placebo's two options are identical and the asymmetry is itself the positional signal, and the trinary independence menus use the analogous correction described below.

Table~\ref{tab:diagnostic-battery} summarizes the diagnostic battery by utility axiom, reporting the menu construction used to evaluate each measure and the corresponding open-weight prompt-level observations. Completeness differs from the other diagnostics because it is evaluated on the union of administered binary prompts rather than on a separate menu class. For each open-weight prompt we extract the full logit vector from a single forward pass and compute the completeness index $\mathcal{C}(x,\tau)$ from Measure~\ref{meas:completeness-index} as the softmax mass on the valid label tokens $\{A,B\}$. For frontier models, completeness is the fraction of sampled responses that parse as a valid label.

\begin{table}[htbp]
  \centering
  \footnotesize
  \renewcommand{\arraystretch}{1.15}
  \begin{tabular}{@{}>{\raggedright\arraybackslash}p{0.37\textwidth}>{\raggedright\arraybackslash}p{0.38\textwidth}r@{}}
    \toprule
    Axiom & Menu Construction & Observations\\
    \midrule
    Completeness (Measure~\ref{meas:completeness-index}) & Union of administered binary menus & 7{,}786 \\
    Reflexivity (Measure~\ref{meas:reflexivity-index}) & Placebo binary menus pairing each bundle with itself & 170 \\
    Monotonicity (Measure~\ref{meas:monotonicity-index}) & Strict mean-risk dominance binary menus & 680 \\
    Transitivity (Measure~\ref{meas:transitivity-index}) & Triad-leg binary menus from three-bundle chains & 1{,}080 \\
    Continuity (Measure~\ref{meas:continuity-index}) & Interpolation-path binary menus with endpoints & 660 \\
    Independence (IIA) (Measure~\ref{meas:iia-index}) & Binary and trinary menus adding a dominated Option $C$ & 930 \\
    \bottomrule
  \end{tabular}
  \caption{Diagnostic battery by utility axiom. Menu observations are open-weight prompt-level counts and include mirrored presentations where used. Transitivity is scored only for \emph{directional} triads, namely those whose first two legs induce a strict chain. Continuity diagnostic conditions on the model's own ranking of the path endpoints.}
  \label{tab:diagnostic-battery}
\end{table}

The diagnostic menus are designed to isolate distinct failures of revealed preference (see Figure \ref{fig:axiom_subplots} for reference). Placebo menus hold the portfolio fixed and therefore attribute any A-versus-B asymmetry to label position. Dominance menus test whether the model respects unambiguous mean-risk improvements. Triad menus test whether pairwise rankings close consistently across three bundles, while interpolation menus ask whether rankings vary continuously along a path between endpoint bundles. Independence menus add a dominated third option to test whether the relative odds between the two focal portfolios are stable to irrelevant menu expansion.

\subsection{Structural Estimation and Implementation}\label{sec:exp-estimation}

Our quadratic utility benchmark identifies preferences from variation in the within-menu differences $\Delta\mu$ and $\Delta R$. Consequently, the estimation menus are selected to maximize the statistical efficiency of these two regressors. After enumerating the admissible binary pairs from the $170$-bundle grid, we exclude identical bundles and pairs in which one option strictly mean-risk dominates the other, and from the remainder we select menus using a greedy forward D-optimal algorithm that sequentially maximizes the determinant of the information matrix over $(\Delta\mu, \Delta R)$ \citep{kiefer1959optimum,mitchell2000algorithm,rose2009constructing}.
\footnote{The D-optimal criterion is invoked here, for the parametric estimation, rather than for the axiom diagnostics, because its purpose is to break the collinearity between the return and risk differences and sharpen identification of the structural slopes.}
The open-weight estimation sample comprises $680$ D-optimal canonical menus together with their $680$ mirrors. The frontier estimation sample uses $400$ canonical estimation menus, administered with mirrors for $800$ total presentations.

For open-weight models, the data for each menu are the two label logits, and estimation proceeds on the position-corrected within-menu logit gap. We fit the first-difference regression in Equation~\eqref{eq:first-difference} with paired corrected differencing, which removes the additive positional advantage $\alpha$ before recovering $\hat{\kappa} = \hat{\theta}_1$ and $\hat{\beta} = -\hat{\theta}_2/\hat{\theta}_1$. Inference for the ratio follows from the Delta method \citep{greene2003econometric}. For frontier models, we observe sampled choices rather than logits and estimate the binary logit index in Equation~\eqref{eq:revealed-prob} by maximizing the Bernoulli log-likelihood in Equation~\eqref{eq:log-likelihood}, recovering $\hat{\kappa}$ and $\hat{\beta}$ from the reduced-form slopes as in Section~\ref{sec:portfolio-choice}. For frontier models, choice probabilities are estimated from repeated queries, with replication counts varying by provider and menu class.\footnote{Structural estimation menus and sampling-based diagnostics receive the highest replication, while Google diagnostics that expose first-token logprobs require only a single query. Some large diagnostic pools are capped for cost control; Appendix~\ref{app:experimental-design} reports the full schedule.}

Design construction and query scheduling use fixed seeds (design and trial seeds set to $42$), and the menu sets are held constant across all subjects. Open-weight logits are extracted deterministically in a single forward pass per menu, so decoding temperature does not enter the open-weight identifying equation. Each query is issued in an isolated session to prevent cross-prompt contamination, and we assess sensitivity to prompt wording through the prompt-variant robustness checks reported in Appendix~\ref{app:robustness}.

For the open-weight observed-logit estimator we report heteroskedasticity-robust standard errors and treat the residual structure as a specification diagnostic, since the logits are deterministic given the prompt and any dispersion reflects departures from the maintained affine specification rather than sampling noise. For the frontier sampled-choice estimator, by contrast, the repeated draws carry a genuine sampling interpretation, and inference follows from the likelihood in the usual way.

\section{Results}\label{sec:results}

We organize the results in two parts. We first evaluate each model against the six axiom diagnostics defined in Section~\ref{sec:structural-model} (completeness, reflexivity, monotonicity, transitivity, continuity, and independence of irrelevant alternatives) using the corresponding indices $\mathcal{C}$, $\mathcal{R}$, $\mathcal{M}$, $\mathcal{T}$, $\mathcal{K}_G$, and $\mathcal{I}$ computed from the observed logit vectors, and report analogous sampled-choice diagnostics for the frontier runs where available. These diagnostics establish whether each model's single-token choices satisfy the conditions for a stable utility-theoretic interpretation before any structural parameters are estimated. We then turn to the structural estimation of the risk preference parameter $\beta$ and the scale parameter $\kappa$ under the maintained benchmark utility $V_i = \mu_i - \beta R_i$.

\subsection{Rationality Diagnostics}

Table~\ref{tab:rationality-diagnostics} reports the model-level averages of each diagnostic index across the open-weight and frontier diagnostic runs. Four patterns stand out. First, completeness is near-perfect for most models. Qwen~3, Gemma~4, and Ministral~3 concentrate effectively all softmax mass on the intended label tokens, with $\mathcal{C} > 0.999$. The Llama~3.1 family is the exception, with the 70B variant allocating roughly one quarter of its probability mass to tokens outside the menu.\footnote{This result is robust to variations in the prompting of the model. For example, some models are trained to produce output in markdown format, which can lead to leading characters for emphasis (e.g. ``*'' or other symbols in front of the labels).} Second, reflexivity reveals substantial heterogeneity. Models that achieve high completeness do not necessarily treat identical portfolios symmetrically. Qwen~3 (14B) and Gemma~4 exhibit near-zero reflexivity. Llama~3.1 (8B Instruct) shows the strongest adherence to the reflexivity axiom among the base models, achieving $\mathcal{R} = 0.836$. All but one of the open weight models we tested (Gemma 4) expressed varying degrees of preference for the first option in the menu, which had to be controlled for in the structural estimation procedure that follows. It is only when the portfolios are exactly the same that the models indulge this position bias, thus there is no economic consequence for this violation of the reflexivity axiom because, regardless of the label that is chosen, the final allocation is the same. See Table \ref{tab:position-bias-difficulty} in the appendix for more details on this result.  Third, monotonicity is satisfied at ceiling across the entire subject pool, and continuity is satisfied at ceiling wherever the premise-verified continuity index is available.

\begin{table}[htbp]
  \centering
  \small
  \resizebox{\textwidth}{!}{
    \begin{tabular}{lcccccc}
      \toprule
      Model & Complete & Reflexive & Continuity & Monotonicity & Transitive & IIA \\
      \midrule

      Qwen 3 (14B) & 1.000 & 0.001 & 1.000 & 1.000 & 0.983 & 0.322 \\
      Qwen 3 (8B) & 1.000 & 0.014 & 1.000 & 1.000 & 0.957 & 0.644 \\
      Gemma 4 (31B) & 1.000 & 0.000 & 1.000 & 1.000 & 0.968 & 0.122 \\
      Ministral 3 (8B) & 1.000 & 0.088 & 1.000 & 1.000 & 0.994 & 0.824 \\
      Gemma 2 (9B) & 0.996 & 0.797 & 1.000 & 1.000 & 0.986 & 0.327 \\
      Llama 3.1 (8B) & 0.902 & 0.836 & 1.000 & 1.000 & 0.994 & 0.920 \\
      Llama 3.1 (70B) & 0.758 & 0.269 & 1.000 & 1.000 & 1.000 & 0.763 \\
      \midrule
      Claude 4.5 Haiku & 0.970 & 0.072 & 1.000 & 1.000 & 0.857 & 0.498 \\
      Claude 4.6 Sonnet & 0.964 & 0.000 & 1.000 & 1.000 & 1.000 & 0.503 \\
      GPT-4o & 0.998 & 0.009 & 1.000 & 1.000 & 1.000 & 0.301 \\
      GPT-4o Mini & 1.000 & 0.008 & 1.000 & 1.000 & 1.000 & 0.586 \\
      Gemini 2.5 Flash Lite & 1.000 & 0.000 & 1.000 & 1.000 & 0.987 & 0.802 \\

      \bottomrule
    \end{tabular}
  }
  \caption{Rationality diagnostics for open-weight and frontier models}
  \label{tab:rationality-diagnostics}
\end{table}

These results indicate that the models uniformly respect stochastic dominance and exhibit smooth, monotone movement of the log-odds along interpolation paths. Transitivity is high throughout, ranging from $0.957$ for Qwen~3 (8B) to $1.000$ for Llama~3.1 (70B) and the OpenAI frontier models, implying that preference cycles are rare but not entirely absent among directional triads.

Fourth, independence of irrelevant alternatives is the weakest axiom across the entire subject pool. No model approaches the ceiling scores observed for monotonicity or continuity. The index $\mathcal{I}$ ranges from $0.122$ for Gemma~4 to $0.920$ for Llama~3.1 (8B Instruct), indicating that adding a strictly dominated third option systematically shifts the log-odds between the original pair. The IIA scores covary strongly with reflexivity. Models exhibiting severe positional label bias, such as Gemma~4 ($\mathcal{R} = 0.000$, $\mathcal{I} = 0.122$) and Qwen~3 14B ($\mathcal{R} = 0.001$, $\mathcal{I} = 0.322$), show the largest violations. This suggests that the same positional sensitivity that prevents symmetric treatment of identical options also makes the pairwise odds fragile to menu expansion.

\bigskip
\noindent\textbf{Finding 1:} The open-weight and frontier models in our sample largely adhere to our measures on the axioms of utility theory. However, we find that a position bias emerges among nearly identical options which leads to violations in reflexivity. Models also do not adhere to the independence of irrelevant alternatives assumption.
\smallskip

\subsection{Structural Estimation}

Having established which axiom conditions each model satisfies, we now estimate the structural parameters of the maintained quadratic expected utility $V_i = \mu_i - \beta R_i$. Table~\ref{tab:structural-estimates} reports estimates of the scale parameter $\hat{\kappa}$ and the risk aversion parameter $\hat{\beta}$. The open-weight rows use the observed-logit estimator in Equation~\eqref{eq:first-difference}, with paired corrected differencing that removes the additive positional bias $\alpha$ before fitting. The frontier rows use the sampled-choice maximum likelihood estimator. Heteroskedasticity-robust standard errors appear in parentheses, and inference for $\hat{\beta}$ follows from the Delta method applied to the ratio $-\hat{\theta}_2 / \hat{\theta}_1$.

\begin{table}[htbp]
  \centering
  \small
  \begin{tabular}{lccc}
    \toprule
    Model & $\hat{\kappa}$ & $\hat{\beta}$ & (Pseudo) R$^2$ \\
    \midrule

    Qwen 3 (14B) & 0.632 & 0.003 & 0.936 \\
    & (0.006) & (0.000) & \\
    Qwen 3 (8B) & 0.311 & 0.003 & 0.889 \\
    & (0.004) & (0.000) & \\
    Gemma 4 (31B) & 1.482 & 0.004 & 0.977 \\
    & (0.008) & (0.000) & \\
    Ministral 3 (8B) & 0.098 & 0.011 & 0.897 \\
    & (0.004) & (0.000) & \\
    Gemma 2 (9B) & 0.249 & 0.008 & 0.911 \\
    & (0.006) & (0.000) & \\
    Llama 3.1 (8B) & 0.039 & 0.009 & 0.886 \\
    & (0.002) & (0.000) & \\
    Llama 3.1 (70B) & 0.162 & 0.005 & 0.939 \\
    & (0.002) & (0.000) & \\
    \midrule
    Claude 4.5 Haiku & 0.393 & 0.004 & 0.634 \\
    & (0.006) & (0.000) & \\
    Claude 4.6 Sonnet & 0.284 & 0.005 & 0.425 \\
    & (0.006) & (0.000) & \\
    GPT-4o & 0.381 & 0.007 & 0.822 \\
    & (0.046) & (0.000) & \\
    GPT-4o Mini & 0.488 & 0.004 & 0.750 \\
    & (0.058) & (0.000) & \\
    Gemini 2.5 Flash Lite & 1.044 & 0.004 & 0.886 \\
    & (0.192) & (0.000) & \\
    \bottomrule
  \end{tabular}
  \caption{Structural estimates of scale $\kappa$ and risk aversion $\beta$. $R^2$ is reported for observed-logit open-weight estimates and pseudo-$R^2$ for sampled-choice frontier estimates.}
  \label{tab:structural-estimates}
\end{table}

Every model in the subject pool yields a positive and statistically significant estimate of $\beta$, indicating that all models exhibit risk aversion. The point estimates vary by a factor of nearly four, ranging from $\hat{\beta} = 0.0031$ for Qwen~3 (8B) to $\hat{\beta} = 0.0114$ for Ministral~3 (8B Instruct). This range implies economically meaningful variation in risk attitudes across model families. Models at the upper end penalize risk more aggressively and will forgo higher expected returns in favor of lower dispersion at comparatively modest risk differentials.

Every model yields a strictly positive estimate for the scale parameter $\hat{\kappa}$, confirming that they naturally interpret the value ordering of the portfolios in the experiment by preferring higher expected returns. If, in contrast, $\kappa$ estimates had been negative then models would have preferred portfolios with lower expected returns, \emph{ceteris paribus}. The magnitude of $\hat{\kappa}$ exhibits wide cross-model variation. Gemma~4 (31B) produces the largest estimate at $\hat{\kappa} = 1.482$, implying that a one-unit increase in the cardinal utility index moves the pre-softmax logit by nearly 1.5 units. By contrast, Llama~3.1 (8B) yields $\hat{\kappa} = 0.039$, roughly forty times smaller.

To provide a richer visualization of these preferences beyond the point estimates in Table~\ref{tab:structural-estimates}, we plot each model's structural $\hat{\beta}$ estimate (except Ministral) in the $\mu$--$\sigma$ plane. Starting from a base bundle with expected return $\mu = 10$ and risk $\sigma = 30$, we sweep over the dense grid of alternative portfolios from our experimental protocol and record the position-corrected logit gap between each grid portfolio and the base. The yellow dashed line overlays the indifference curve implied by the mean-variance structural estimates, and the heatmap colors encode the sign and magnitude of the logit difference, with blue regions preferred to the base and red regions dispreferred.

\begin{figure}[htbp]
  \centering
  \begin{subfigure}{0.43\textwidth}
    \includegraphics[width=\linewidth]{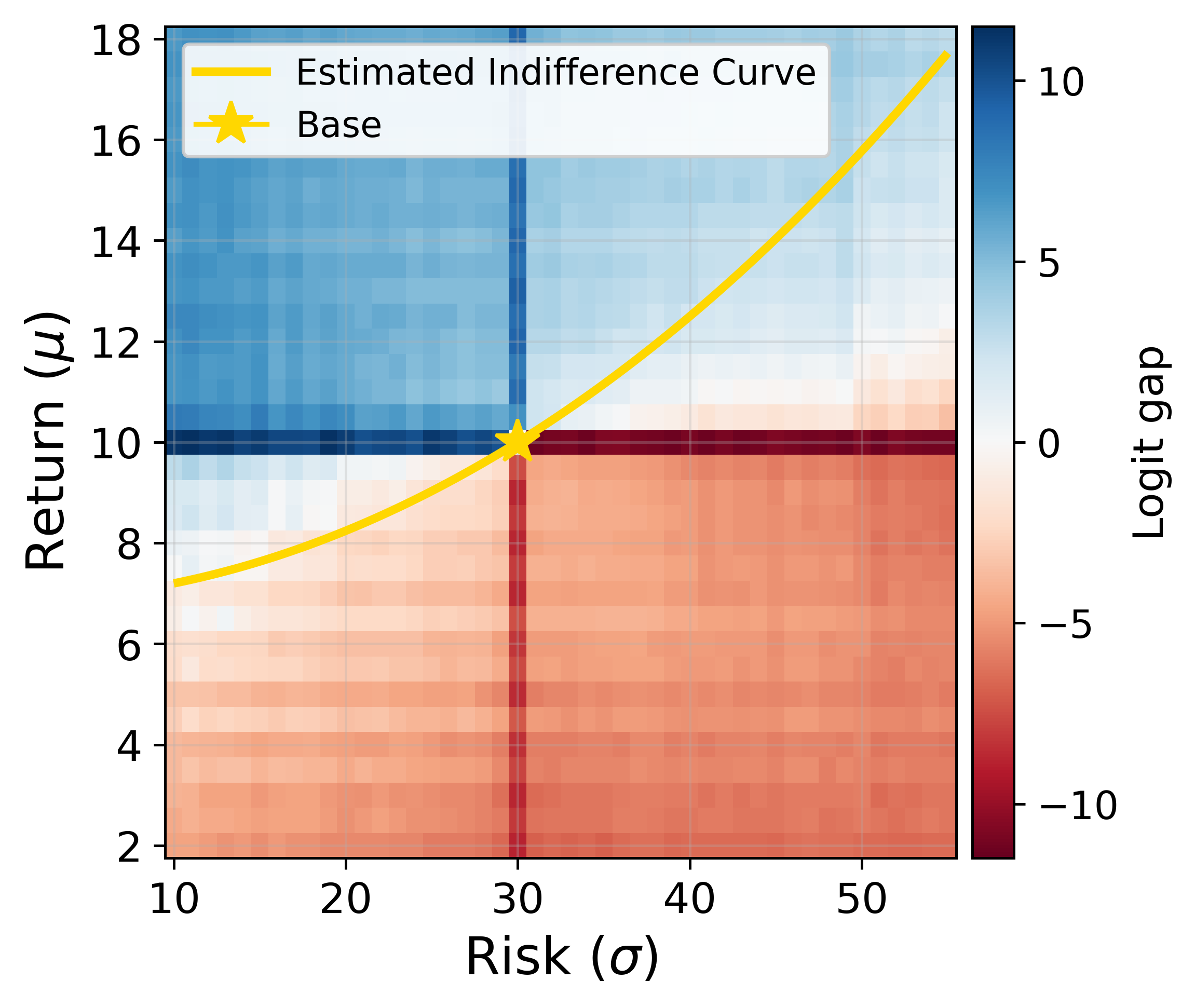}
    \caption{Qwen 3 (14B)}
  \end{subfigure}\hfill
  \begin{subfigure}{0.43\textwidth}
    \includegraphics[width=\linewidth]{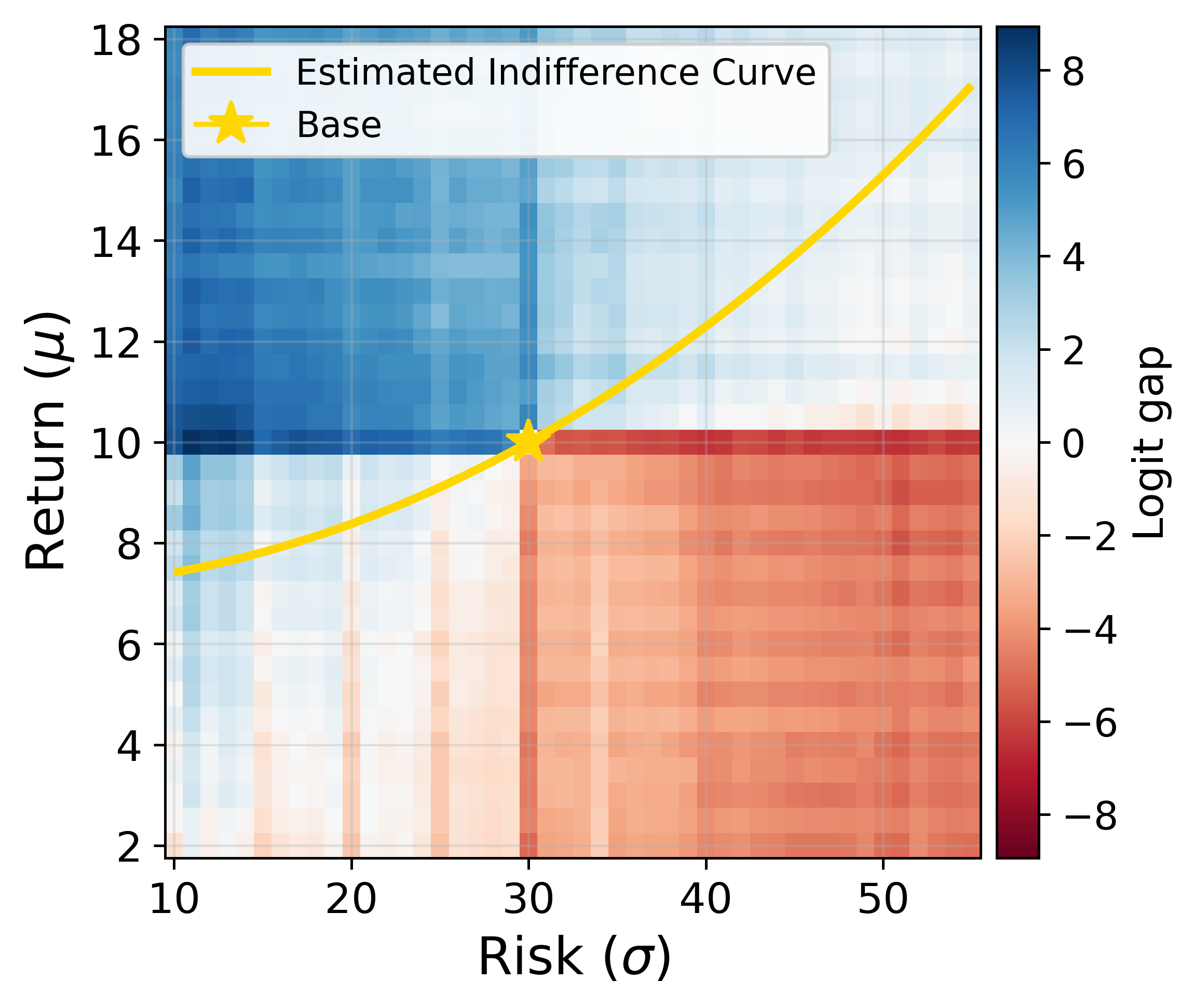}
    \caption{Qwen 3 (8B)}
  \end{subfigure}

  \vspace{0.5em}
  \begin{subfigure}{0.43\textwidth}
    \includegraphics[width=\linewidth]{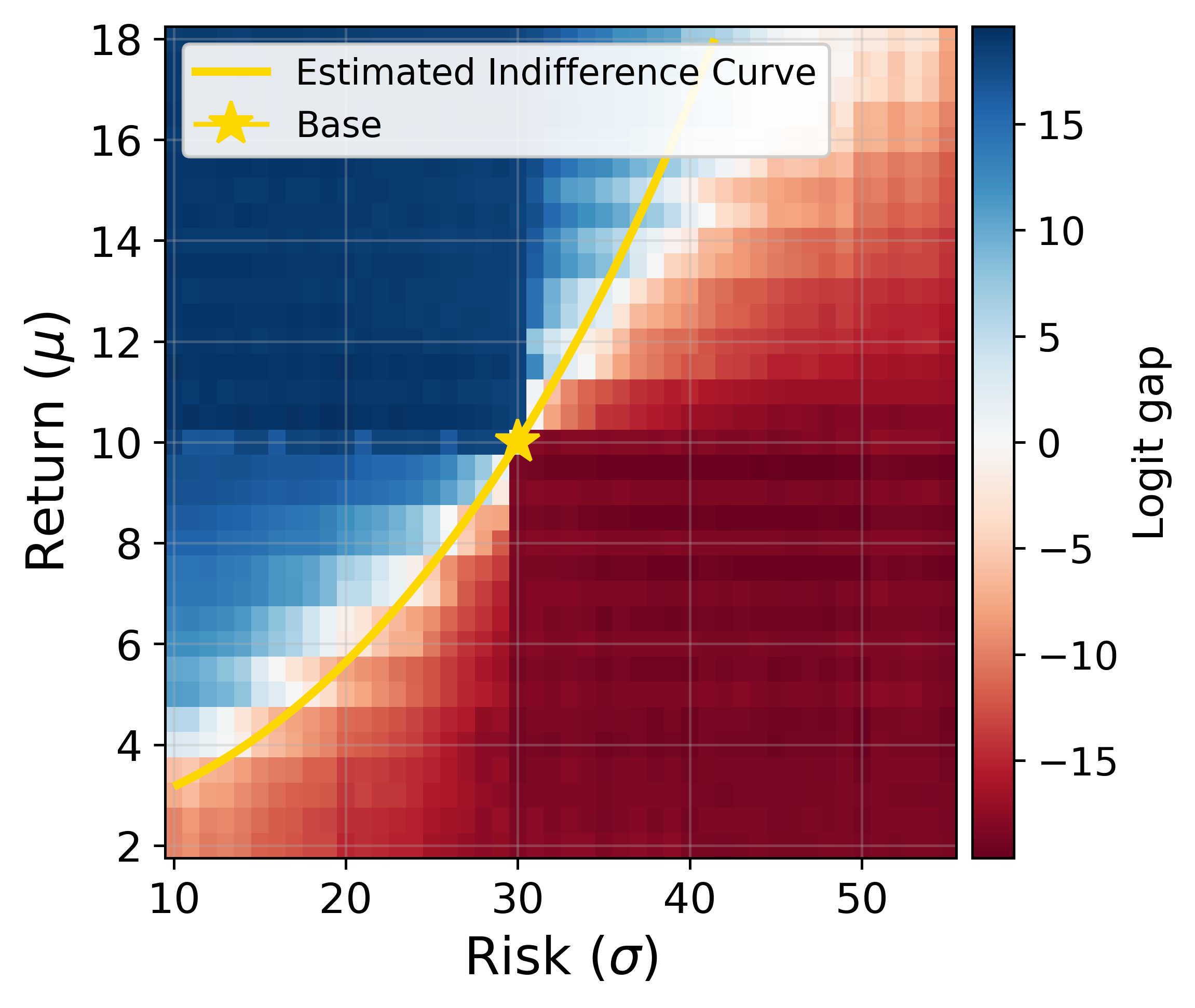}
    \caption{Gemma 4 (31B Instruct)}
  \end{subfigure}\hfill
  \begin{subfigure}{0.43\textwidth}
    \includegraphics[width=\linewidth]{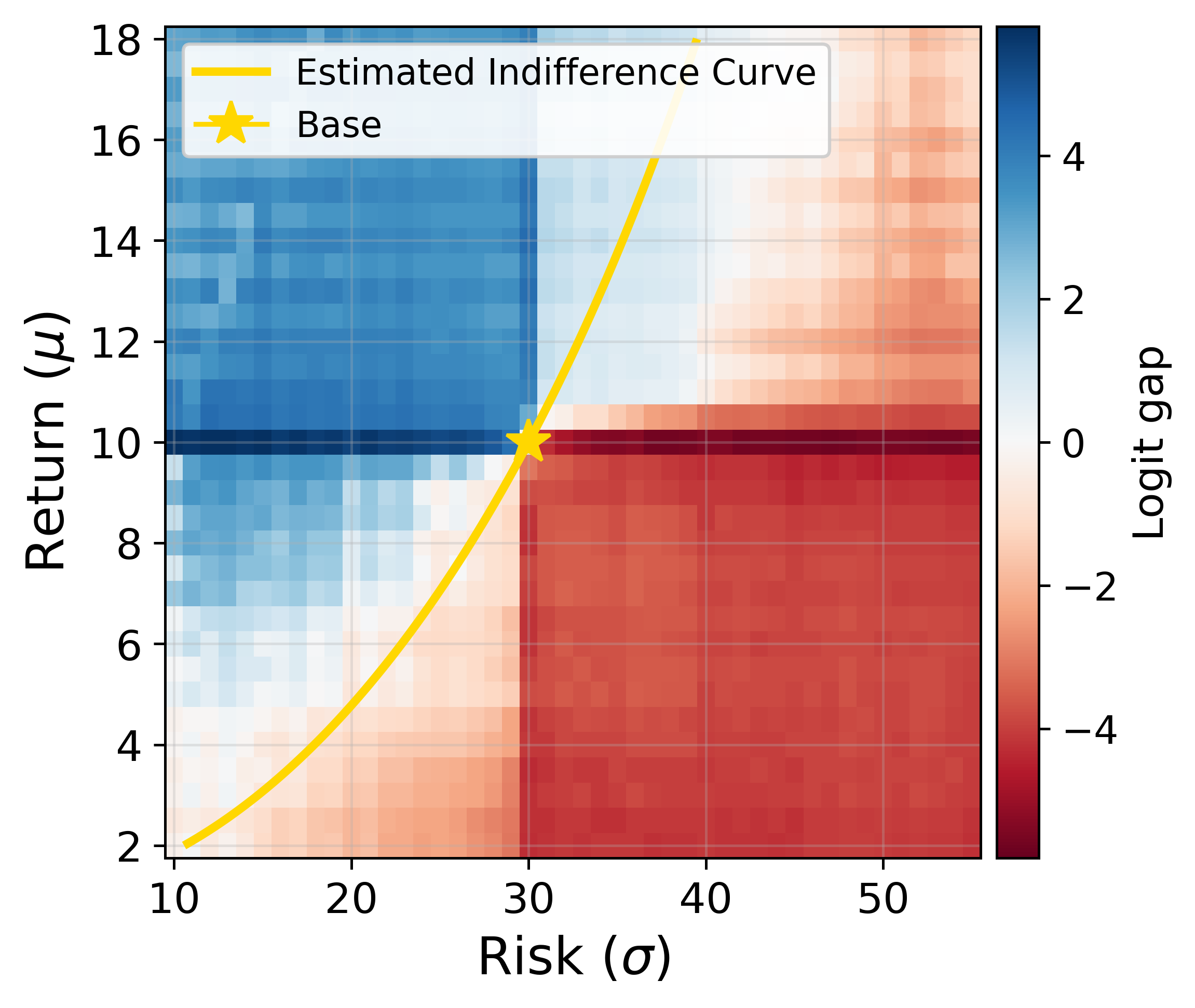}
    \caption{Gemma 2 (9B Instruct)}
  \end{subfigure}

  \vspace{0.5em}
  \begin{subfigure}{0.43\textwidth}
    \includegraphics[width=\linewidth]{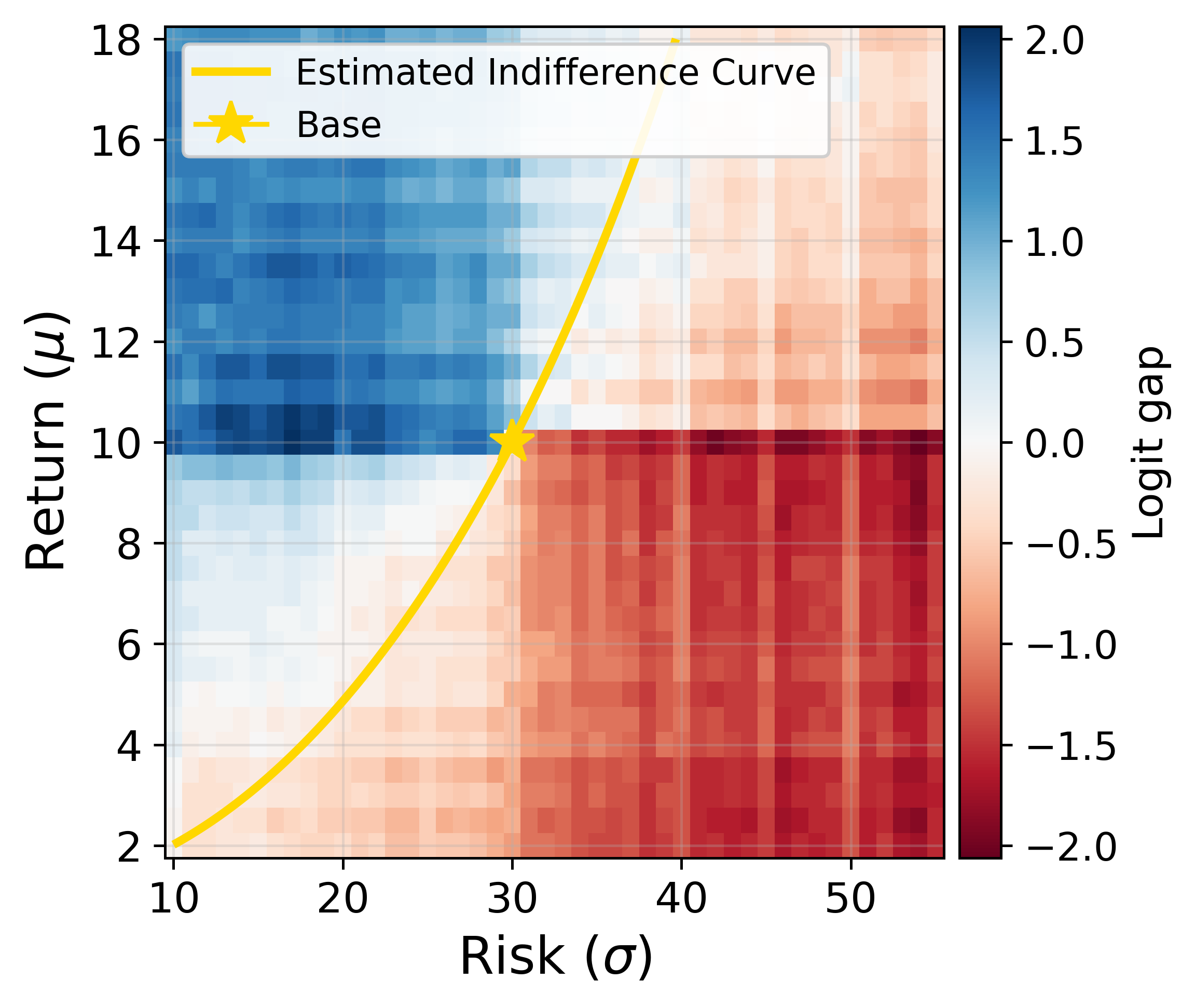}
    \caption{Llama 3.1 (8B Instruct)}
  \end{subfigure}\hfill
  \begin{subfigure}{0.43\textwidth}
    \includegraphics[width=\linewidth]{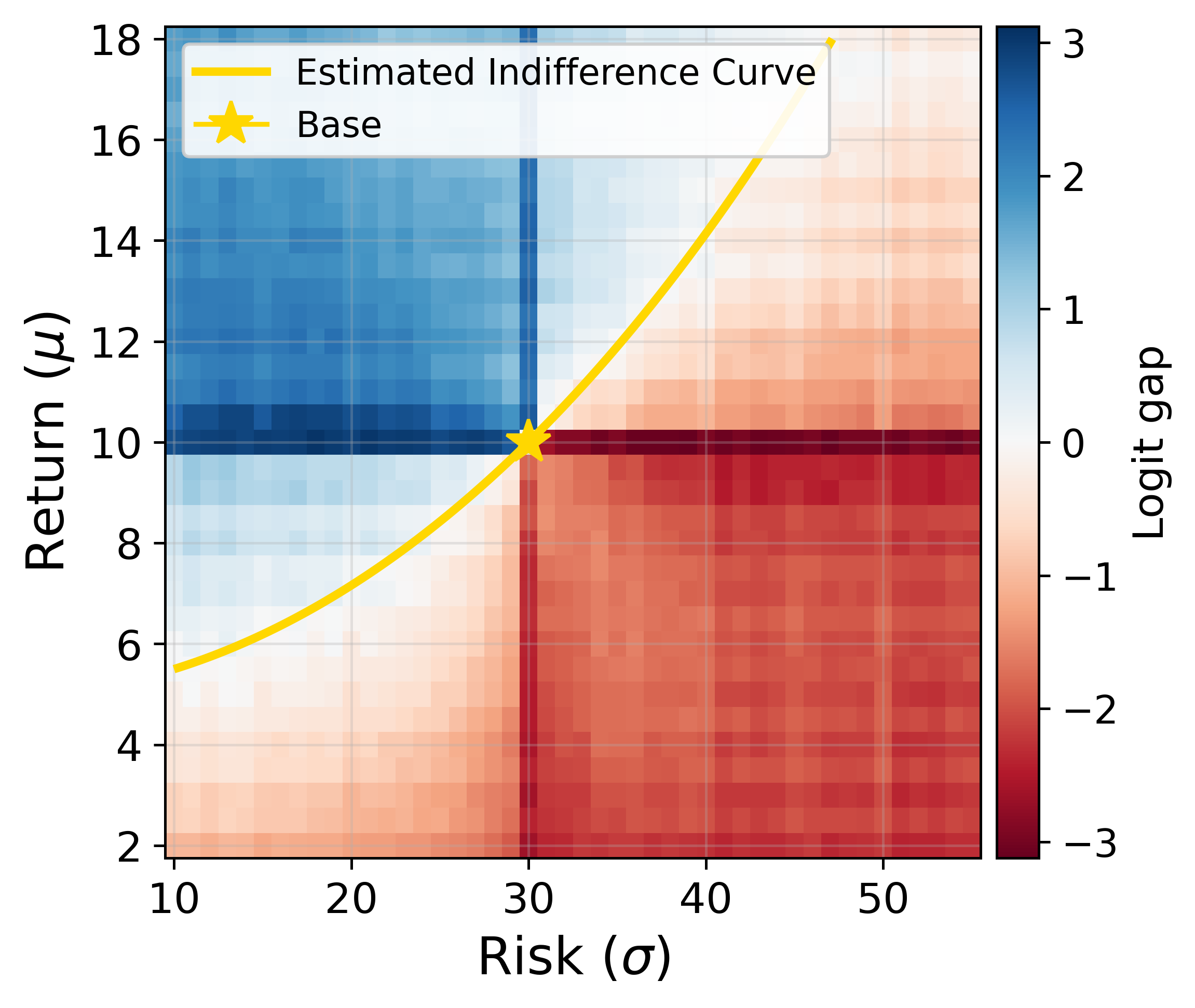}
    \caption{Llama 3.1 (70B Instruct)}
  \end{subfigure}

  \caption{Empirical indifference curves for open-weight models mapped over the portfolio space.}
  \label{fig:indifference-curves}
\end{figure}

Several patterns emerge from these plots. All six models produce upward-sloping indifference curves, confirming that higher risk must be compensated by higher expected return. The curves are roughly convex in the $\mu$--$\sigma$ plane, consistent with the quadratic utility specification maintained in the structural estimation. However, the steepness of the curves varies considerably across models, reflecting the heterogeneity in $\hat{\beta}$ documented above. Models with larger risk aversion estimates, such as Gemma~2 (9B), trace out steeper curves that demand more return compensation per unit of additional risk. Models with smaller $\hat{\beta}$, such as Qwen~3 (14B), produce flatter curves that tolerate substantially more risk for modest return increments.

\bigskip
\noindent\textbf{Finding 2:} All open-weight and frontier models in our sample exhibit varying degrees of risk aversion that can be approximated
by a quadratic utility function.
\smallskip

\subsection{Inducing Preferences in Language Models}\label{subsec:fine-tuning}

Can specific economic preferences be induced into the language model? This question matters because a principal who delegates an allocation decision to a language model delegates authority over the tradeoff that the instruction leaves unresolved. If the model's default risk attitude differs from the principal's, then the delegated choice implements the model's objective rather than the principal's. The preceding results show that these default preferences can be measured. We now ask whether the preference-misalignment component of this agency problem can also be addressed directly by inducing the model to reveal a target economic preference in subsequent choices.

We study this question by fine-tuning Llama~3.1 (8B Instruct) toward two principal-specified risk preferences. The principal's desired ordering over portfolios is represented by the same maintained utility index,
\[
  V_i^* = \mu_i - \beta^* R_i,
\]
where $\beta^*$ is chosen by the principal. We consider two target specifications. The first is approximately risk neutral, with $\beta^* = 0.0001$. The second is more risk averse, with $\beta^* = 0.0045$. These targets are not prompts asking the model to describe a preference. They are the preferences the delegated agent is trained to implement in its own choice probabilities.\footnote{This differs from prompt engineering. Prompt engineering conditions the model on instructions supplied in the input context and leaves the model weights unchanged, so any induced behavior persists only while that context is present. Fine-tuning updates the weights that generate the choice rule itself, allowing the principal-specified preference to be embedded in the model's revealed choice probabilities.}

For each training menu, we compute the utility difference implied by the principal's target preference and translate it into the log odds that the delegated model should assign to the two options.
\begin{equation}\label{eq:principal-target-odds}
  \Delta z_m^*
  \equiv
  \log\frac{P^*(A \mid m)}{P^*(B \mid m)}
  =
  \kappa\left(V_A^* - V_B^*\right)
  =
  \kappa\left(\Delta\mu_m - \beta^* \Delta R_m\right).
\end{equation}
For the same menu, let $\Delta z_m(\theta)$ denote the model's actual logit gap between the two answer labels, where $\theta$ collects the model weights. The preference component of fine-tuning chooses weights to close the distance between the actual logit gap and the target logit gap,
\[
  \mathcal{L}(\theta)
  =
  \frac{1}{M}\sum_{m=1}^{M}
  \left(\Delta z_m(\theta) - \Delta z_m^*\right)^2 .
\]
Minimizing $\mathcal{L}(\theta)$ updates the model toward the principal's choice rule. If the model places too much relative logit weight on option $A$, the residual is positive and the objective rewards changes that reduce the $A$ versus $B$ gap. If the model places too little relative logit weight on option $A$, the objective rewards changes that increase that gap. Repeating this comparison across the training menus moves the model's logit surface toward the logit surface implied by the principal's utility. The scale parameter $\kappa$ governs the strength of the probabilistic response but does not determine the recovered risk preference, since $\hat{\beta}$ is identified by the ratio of the risk and return slopes. Appendix~\ref{app:fine-tuning} gives the formal multi-term objective, the training corpus, the logit-scale choice, and the implementation details.

After fine-tuning, we evaluate the delegated models exactly as before, on held-out menus not used in training. The test is whether the model now reveals the principal's target preference rather than the base model's default preference, and whether that induced preference passes the same diagnostic battery. Table~\ref{tab:rationality-diagnostics-fine-tuned} reports the rationality diagnostics. Relative to the baseline Llama 3.1 (8B Instruct), both induced-preference models substantially improve completeness and reflexivity, with completeness above $0.99$ and reflexivity above $0.98$. Monotonicity remains perfect in both cases. The near risk-neutral target achieves perfect transitivity and raises IIA to $0.9484$, while the more risk-averse target maintains high diagnostic performance, with IIA equal to $0.8000$ and transitivity equal to $0.9583$.

\begin{table}[htbp]
  \centering
  \small
  \resizebox{\textwidth}{!}{
    \begin{tabular}{lcccccc}
      \toprule
      Model & Complete & Reflexive & Continuity & Monotonicity & Transitive & IIA \\
      \midrule
      Llama 3.1 (8B, $\beta=0.0001$) & 0.9961 & 0.9941 & 0.9333 & 1.0000 & 1.0000 & 0.9484 \\
      Llama 3.1 (8B, $\beta=0.0045$) & 0.9914 & 0.9870 & 1.0000 & 1.0000 & 0.9583 & 0.8000 \\
      \bottomrule
    \end{tabular}
  }
  \caption{Rationality diagnostics for fine-tuned models}
  \label{tab:rationality-diagnostics-fine-tuned}
\end{table}

Table~\ref{tab:structural-estimates-fine-tuned} reports the structural estimates and Figure~\ref{fig:indifference-curves-fine-tuned} plots the implied indifference curves over the portfolio space. The induced models reveal the principal-specified preferences with high precision. The near risk-neutral specification targets $\beta^* = 0.0001$, and the fine-tuned model yields $\hat{\beta} = 0.000087$ with a standard error of $0.000003$. The more risk-averse specification targets $\beta^* = 0.0045$, and the fine-tuned model yields $\hat{\beta} = 0.004560$ with a standard error of $0.000017$. In both cases, the target value lies within the 95\% confidence interval. The scale parameter estimates are similarly precise. The low-$\beta$ model yields $\hat{\kappa} = 1.001301$ with a standard error of $0.000720$, and the high-$\beta$ model yields $\hat{\kappa} = 0.718940$ with a standard error of $0.002206$. The restricted quadratic surface provides a tight fit to the indifference data, producing $R^2$ values of $0.9992$ and $0.9935$.

The economic interpretation is that the principal need not accept the model's default risk attitude as an unavoidable agency cost within this portfolio environment. A principal can specify a target utility function, fine-tune the model toward the choice odds implied by that utility, and then verify out of sample that the delegated model reveals the intended preference. The exercise therefore addresses the preference-misalignment part of the principal-agent problem studied here. The agent's objective can be measured, deliberately shifted, and audited with the same structural tools.

\begin{table}[htbp]
  \centering
  \small
  \begin{tabular}{lccc}
    \toprule
    Model & $\hat{\kappa}$ & $\hat{\beta}$ & R$^2$ \\
    \midrule
    Llama 3.1 (8B, $\beta=0.0001$) & 1.001301 & 0.000087 & 0.999 \\
    & (0.000720) & (0.000003) & \\
    Llama 3.1 (8B, $\beta=0.0045$) & 0.718940 & 0.004560 & 0.994 \\
    & (0.002206) & (0.000017) & \\
    \bottomrule
  \end{tabular}
  \caption{Structural estimates of scale $\kappa$ and risk aversion $\beta$ for fine-tuned models.}
  \label{tab:structural-estimates-fine-tuned}
\end{table}

\begin{figure}[htbp]
  \centering
  \begin{subfigure}{0.43\textwidth}
    \includegraphics[width=\linewidth]{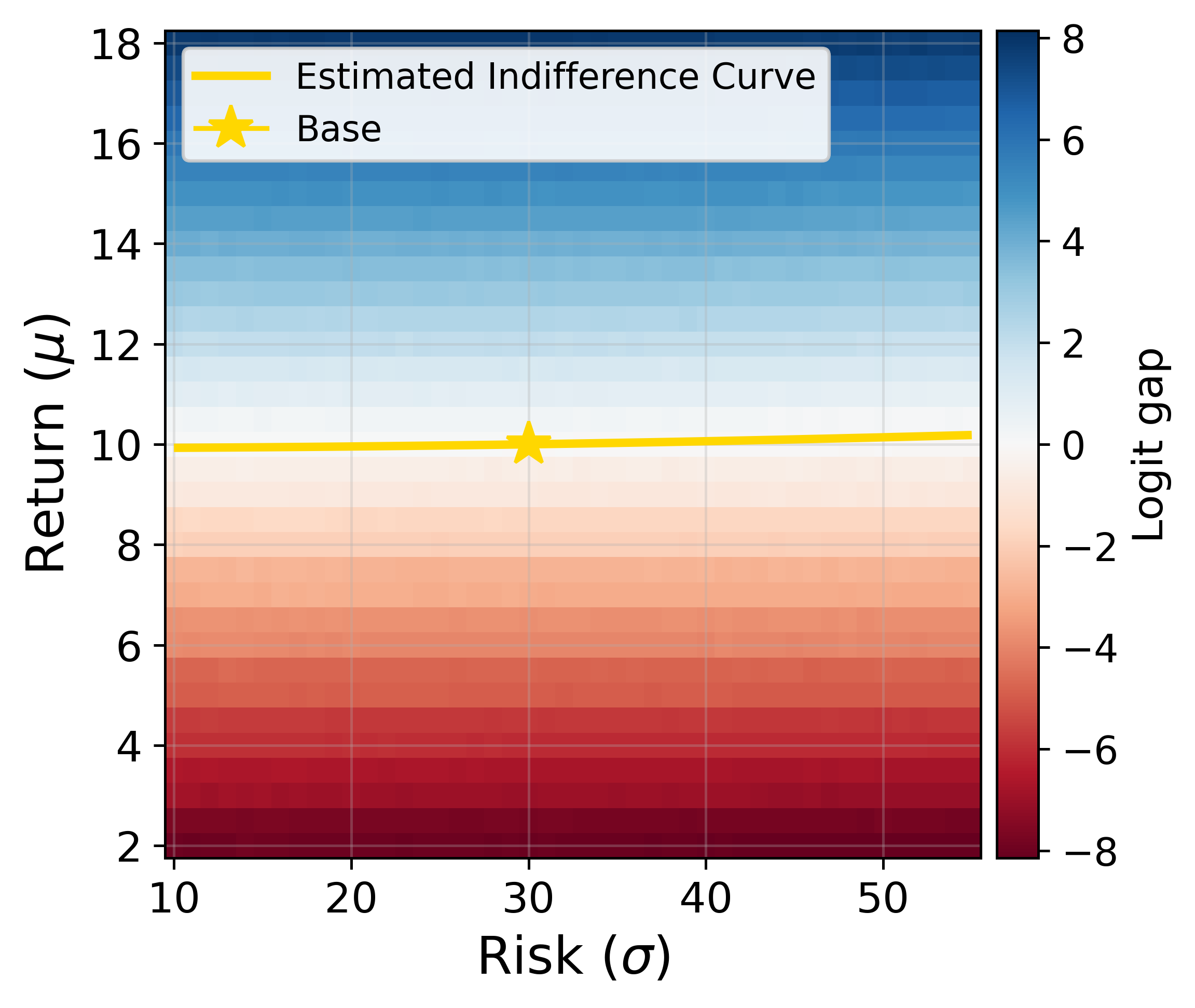}
    \caption{Llama 3.1 (8B $\beta=0.0001$)}
  \end{subfigure}\hfill
  \begin{subfigure}{0.43\textwidth}
    \includegraphics[width=\linewidth]{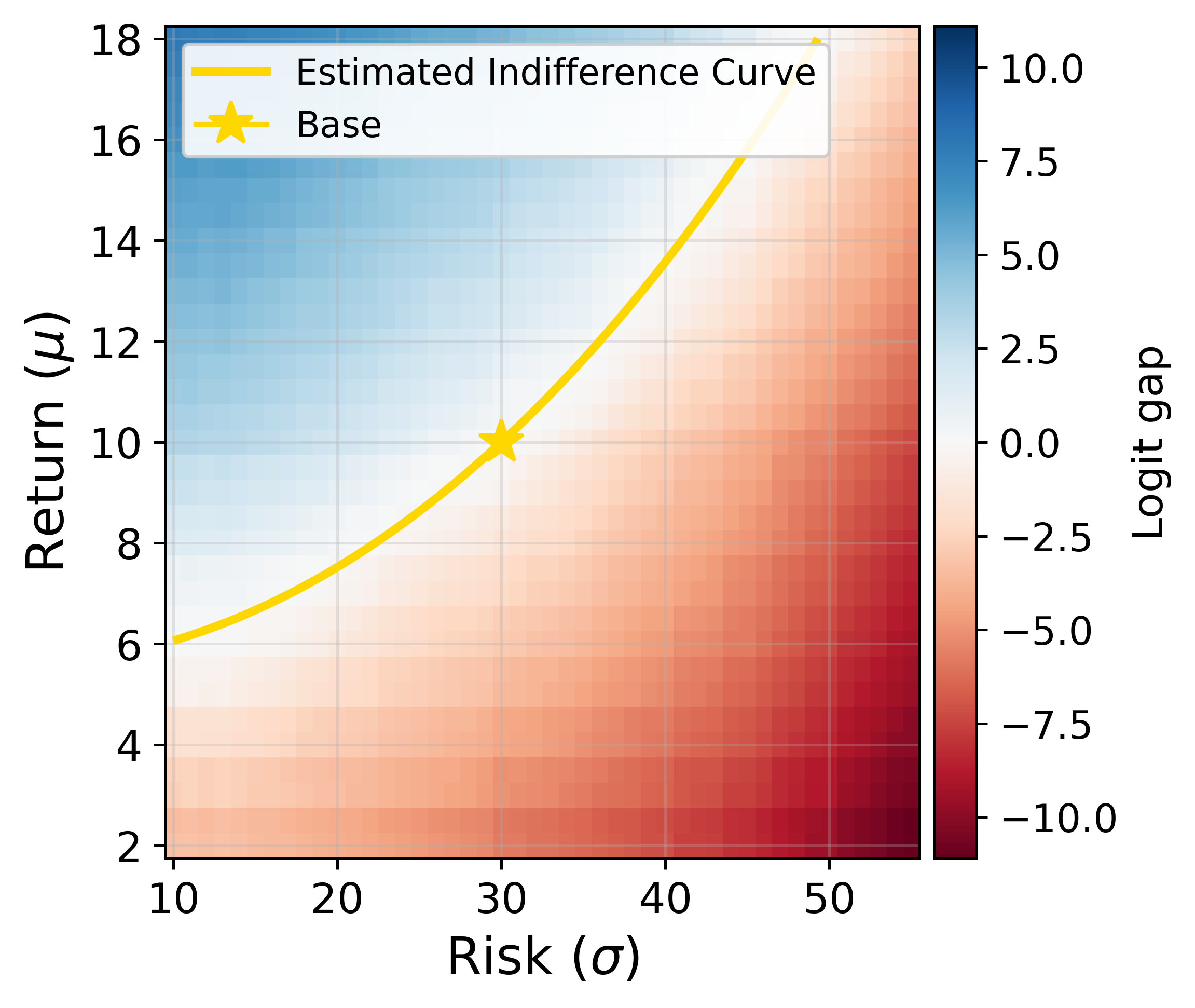}
    \caption{Llama 3.1 (8B $\beta=0.0045$)}
  \end{subfigure}
  \caption{Empirical indifference curves for fine-tuned models mapped over the portfolio space.}
  \label{fig:indifference-curves-fine-tuned}
\end{figure}

\bigskip
\noindent\textbf{Finding 3:}  A principal can install its preferred risk attitude in a language model through fine-tuning, embedding the preference in the model's weights so that it governs how the model resolves tradeoffs.
\smallskip

\section{Conclusion}\label{sec:conclusion}

The starting point of this paper is an equivalence. When a language model is required to state its choice with a single token, the rule by which it selects that token is exactly the conditional logit random utility model of \cite{mcfadden1972conditional}. Using this equivalence, we develop a structural methodology for recovering the economic preferences of these models directly from their generative outputs. The raw pre-softmax logits assigned to menu labels can therefore be interpreted as a systematic utility index and place the choices of a language model directly within the scope of standard structural discrete choice methods. It yields one estimator when the model's internal scores are observable, as with open-weight models, and the standard maximum likelihood estimator when only realized choices are observed, as with proprietary frontier models (i.e. ChatGPT). The same framework also motivates six revealed-preference diagnostics for completeness, reflexivity, monotonicity, transitivity, continuity, and independence of irrelevant alternatives. These diagnostics allow the stability of the utility interpretation to be evaluated before any structural parameter is estimated.

Applying this method to analyze the choices of several language models, our results reveal a clear distinction between coherence about economic content and invariance to economically irrelevant features. Across the models we study, monotonicity holds for every model, continuity holds wherever the premise-verified index is available, and transitivity is high. The models therefore respect mean-risk dominance, move smoothly along the interpolation paths in the experiment, and rarely generate cycles across separately elicited menus. In this sense, the models behave like textbook rational agents. The main failures arise from features of the elicitation environment that should be economically irrelevant. Reflexivity frequently fails because models assign different weights to identical portfolios upon label or position permutation. Independence of irrelevant alternatives proves to be the weakest diagnostic, as introducing a dominated option routinely alters the pairwise odds between the original alternatives. These failures covary, which suggests that positional sensitivity and menu sensitivity are closely related in single-token model choice.

Conditional on these limits, the structural estimates recover economically interpretable risk preferences. Every open-weight and frontier model in our sample is risk averse, with a positive and statistically significant risk-aversion parameter under the maintained quadratic specification. The estimates vary by nearly a factor of four across model families. This variation is large enough to imply different portfolio rankings and therefore different allocations when a model is asked to choose among risky alternatives. The estimated indifference curves are upward sloping and approximately convex in the mean-variance plane, so the familiar quadratic apparatus provides a useful first-order description of the preference surface these models reveal.

The fine-tuning exercise provides a structural validation of this interpretation. We install a target utility function in an open-source model at two known risk-aversion targets, and our estimator recovers both targets with high precision on menus the model never saw in training. The exercise also shows that the main invariance failures are not fixed properties of the base model. In the controlled environment studied here, the relevant preference object can therefore be measured and deliberately shifted, and the same intervention can improve the invariance properties required for its interpretation.

Two boundaries delimit what we claim. Each choice is a single token produced in one forward pass rather than the outcome of extended deliberation, and the domain is a controlled portfolio environment rather than the full space of economic decisions. These boundaries are what make the measurement exact, since within them the mapping from logits to utility follows from the softmax rule rather than from an auxiliary assumption. They also locate the object precisely. What we recover is the default preference a model brings before it reasons, which is the relevant object when a principal delegates under an incomplete instruction and the model answers directly. Whether that default persists when the model deliberates across many tokens, when the prompt is enriched toward a fuller mandate, or when the domain moves beyond portfolio choice is the natural next question, as is whether the same preference governs the longer chains of reasoning through which these models increasingly act. While we briefly entertain these extensions in our robustness checks, we treat the single-token measurement as our primary interest and the foundation those extensions build on. Whether structurally measured preferences persist under deliberation and across economic environments is the natural next question for this line of work.

These findings matter for the use of language models as economic agents. Delegating allocation decisions to a model also delegates authority to the preferences that govern its choices. Those preferences need not be observed by the principal, and the diagnostics in this paper show that they need not be invariant to irrelevant features of the menu. A coherent deployment strategy therefore requires more than evidence that a model can correctly rank dominated alternatives. It requires precise measurement of the preference governing choice. It requires diagnostics to verify the stability of that preference. Finally, it requires tools to shift the preference when a different objective is intended. The evidence presented here supplies a controlled demonstration of exactly that program.

\clearpage
\appendix

\section{Proofs}\label{app:proofs}

\subsection{Proof of Proposition \ref{prop:rum}}

\begin{proof}
  Assume the error terms $\epsilon_i$ are independent and identically distributed following a standard Gumbel (Type I Extreme Value) distribution. The cumulative distribution function (CDF) is $F(\epsilon) = \exp(-\exp(-\epsilon))$ and the probability density function (PDF) is $f(\epsilon) = \exp(-\epsilon)\exp(-\exp(-\epsilon))$. The model reveals a preference for token $i$ if it provides the highest utility among all alternatives in the vocabulary $\mathcal{V}$. That is, $U_i > U_j$ for all $j \neq i$. The probability of choosing token $i$ is given by
  $$P(i \mid x, \tau) = P(u_i(x) + \tau \epsilon_i > u_j(x) + \tau \epsilon_j \quad \forall j \neq i).$$

  Rearranging to isolate $\epsilon_j$ yields
  $$P(i \mid x, \tau) = P\left(\epsilon_j < \frac{u_i(x) - u_j(x)}{\tau} + \epsilon_i \quad \forall j \neq i\right).$$

  Conditional on a specific value of $\epsilon_i$, the probability that this inequality holds for all $j \neq i$ is the product of their individual CDFs, due to independence
  $$P(i \mid x, \tau, \epsilon_i) = \prod_{j \neq i} \exp\left(-\exp\left(-\left(\frac{u_i(x) - u_j(x)}{\tau} + \epsilon_i\right)\right)\right).$$

  To find the unconditional probability, we integrate this product over the marginal distribution of $\epsilon_i$ as follows
  $$P(i \mid x, \tau) = \int_{-\infty}^{\infty} f(\epsilon_i) \prod_{j \neq i} F\left(\frac{u_i(x) - u_j(x)}{\tau} + \epsilon_i\right) d\epsilon_i.$$

  Substitute the Gumbel PDF for $f(\epsilon_i)$ and separate the exponent terms inside the product
  $$P(i \mid x, \tau) = \int_{-\infty}^{\infty} \exp(-\epsilon_i) \exp(-\exp(-\epsilon_i)) \prod_{j \neq i} \exp\left(-\exp(-\epsilon_i) \exp\left(-\frac{u_i(x) - u_j(x)}{\tau}\right)\right) d\epsilon_i.$$

  We can combine the $\exp(-\exp(-\epsilon_i))$ terms. Notice that the standalone PDF component is exactly the term we would get if $j=i$ inside the product. Therefore, we can absorb it by indexing the product (and subsequent sum) over all $j \in \mathcal{V}$
  $$P(i \mid x, \tau) = \int_{-\infty}^{\infty} \exp(-\epsilon_i) \exp\left(-\exp(-\epsilon_i) \sum_{j \in \mathcal{V}} \exp\left(\frac{u_j(x) - u_i(x)}{\tau}\right)\right) d\epsilon_i.$$

  To solve this integral, apply the substitution $t = \exp(-\epsilon_i)$, which gives $dt = -\exp(-\epsilon_i) d\epsilon_i$. The limits of integration change from $(-\infty, \infty)$ to $(\infty, 0)$. To simplify the notation, let $K = \sum_{j \in \mathcal{V}} \exp\left(\frac{u_j(x) - u_i(x)}{\tau}\right)$. Then
  $$P(i \mid x, \tau) = \int_{\infty}^{0} - \exp(-K t) dt = \int_{0}^{\infty} \exp(-K t) dt.$$

  Evaluating this standard exponential integral yields
  $$P(i \mid x, \tau) = \left[ -\frac{1}{K} \exp(-K t) \right]_0^\infty = \frac{1}{K}.$$

  Substituting the expression for $K$ back into the result gives
  $$P(i \mid x, \tau) = \frac{1}{\sum_{j \in \mathcal{V}} \exp\left(\frac{u_j(x) - u_i(x)}{\tau}\right)} = \frac{\exp(u_i(x) / \tau)}{\sum_{j \in \mathcal{V}} \exp(u_j(x) / \tau)}.$$

  This exactly recovers Equation~\eqref{eq:softmax}, completing the proof.
\end{proof}

\section{Experimental Design and Identification}\label{app:experimental-design}

\subsection{Label Tokens}

The forced-choice protocol uses the labels A and B in binary menus and A, B, and C in three-option independence menus. For open-weight models, grouped label logits and probabilities are computed by aggregating every vocabulary token whose decoded text matches the target label after trimming whitespace, with alphabetic labels matched case-insensitively. This is the rule implemented by \texttt{build\_token\_map} in the local-model evaluation code and reused by the fine-tuning code. Table~\ref{tab:label-token-ids} reports the token IDs resolved from the cached tokenizers used in the run environment. The two fine-tuned Llama variants use the Llama~3.1~(8B) tokenizer and therefore share its label-token map.

\begin{table}[htbp]
  \centering
  \small
  \begin{tabularx}{\textwidth}{lXXX}
    \toprule
    Model tokenizer & A IDs & B IDs & C IDs \\
    \midrule
    Gemma-2-9B & 282, 314, 476, 586, 235250, 235280 & 283, 315, 518, 599, 235268, 235305 & 284, 316, 498, 585, 235260, 235288 \\
    Gemma-4-31B & 303, 335, 496, 562, 236746, 236776 & 304, 336, 518, 603, 236763, 236799 & 305, 337, 505, 565, 236755, 236780 \\
    Llama-3.1-8B and fine-tuned variants & 32, 64, 264, 362, 11575, 23845, 118586 & 33, 65, 293, 426, 2282, 13083, 108693, 115648 & 34, 66, 272, 356, 1470, 6391, 116545, 117155 \\
    Llama-3.1-70B & 32, 64, 264, 362, 11575, 23845, 118586 & 33, 65, 293, 426, 2282, 13083, 108693, 115648 & 34, 66, 272, 356, 1470, 6391, 116545, 117155 \\
    Ministral-3-8B & 1065, 1097, 1261, 1349, 44056, 88974 & 1066, 1098, 1289, 1398, 40796 & 1067, 1099, 1272, 1359, 6622, 94097 \\
    Qwen3-8B & 32, 64, 264, 362, 11323, 22985 & 33, 65, 293, 425, 2233, 12791 & 34, 66, 272, 356, 1444, 6258 \\
    Qwen3-14B & 32, 64, 264, 362, 11323, 22985 & 33, 65, 293, 425, 2233, 12791 & 34, 66, 272, 356, 1444, 6258 \\
    \bottomrule
  \end{tabularx}
  \caption{Open-weight label-token IDs aggregated into menu labels.}
  \label{tab:label-token-ids}
\end{table}

\subsection{Identification Under the Two Observation Regimes}

This appendix next states the binary identification results under the two observation regimes. Let $\mathcal{B}$ denote the finite grid of candidate bundles indexed by observed moments $(\mu,\sigma^2,R)$ with $R = \sigma^2 + \mu^2$. Let $\mathcal{P} \subset \mathcal{B} \times \mathcal{B}$ denote the set of admissible ordered binary menus $(A,B)$ with distinct bundles. For the structural estimation sample used in the paper, this algebraic pool is further restricted by excluding degenerate first-difference pairs and strict mean-risk dominance pairs, as described in Section~\ref{sec:exp-estimation}. For any menu $j \in \mathcal{P}$, write $x_j = (\Delta \mu_j, \Delta R_j)'$ with differences taken in A minus B order. Mirrored menus reverse the sign of $x_j$ and help balance the design, but they do not change its span.

\subsection{Observed-Logit Design}

For observed logits, within-menu differencing removes the menu fixed effect. Pairing each canonical presentation with its mirror removes the additive position component and yields the content signal
\[
  \Delta z_j
  =
  \tfrac{1}{2}
  \left(\Delta u_j-\Delta u_j^{\mathrm{mirror}}\right).
\]
The estimating equation is then
\[
  \Delta z_j = \theta_1 \Delta \mu_j + \theta_2 \Delta R_j + \tilde{\eta}_j,
\]
with no intercept. Define the second-moment matrix
\[
  M_N = \frac{1}{N}\sum_{j=1}^N x_j x_j'.
\]

If $M_N$ is nonsingular, then ordinary least squares on the paired content-signal equation identifies $(\theta_1,\theta_2)$. Under the maintained mapping $\theta_1 = \kappa$ and $\theta_2 = -\kappa\beta$, the same design identifies $\kappa = \theta_1$ and $\beta = -\theta_2/\theta_1$ whenever $\theta_1 \neq 0$.

For observed logits, precision depends on the same second-moment matrix. Under homoskedastic errors, the covariance matrix of $\hat{\theta}$ is proportional to $M_N^{-1}$. This covariance expression is a design benchmark. Reported inference uses heteroskedasticity-robust covariance estimates and the Delta method for the ratio parameter. A natural benchmark is the D-optimal criterion of \cite{louviere1983design},
\[
  \mathcal{J}^{D}
  \in
  \arg\max_{\mathcal{J}\subset\mathcal{P},\ |\mathcal{J}|=N}
  \det\left(\sum_{j\in\mathcal{J}} x_j x_j'\right).
\]
Since
\[
  \det(M_N)
  =
  \overline{\Delta \mu^2}\,\overline{\Delta R^2}
  -
  \overline{\Delta \mu \Delta R}^{\,2},
\]
the criterion favors menus with large spread in both coordinates and little correlation between them. Under a symmetric rectangular approximation to the feasible difference set, the benchmark places equal mass on the four corner points $(\pm 1, \pm 1)$. On the actual finite grid, the same logic calls for menus near the boundary of the feasible difference set, balanced representation across quadrants, and mirrored orderings so that
\[
  \sum_{j=1}^N \Delta \mu_j \approx 0,
  \qquad
  \sum_{j=1}^N \Delta R_j \approx 0,
  \qquad
  \sum_{j=1}^N \Delta \mu_j \Delta R_j \approx 0.
\]
Following \cite{louviere1983design}, a c-optimal refinement after a pilot estimate $\tilde{\beta}$ can target the precision of the ratio parameter by solving
\[
  \mathcal{J}^{\beta}
  \in
  \arg\min_{\mathcal{J}\subset\mathcal{P},\ |\mathcal{J}|=N}
  (\tilde{\beta},\,1)
  \left(\sum_{j\in\mathcal{J}} x_j x_j'\right)^{-1}
  (\tilde{\beta},\,1)'.
\]

\subsection{Sampled-Choice Design}

With sampled choices, the identified index is the binary logit index. The reported frontier estimator includes an intercept that absorbs an average A-label advantage,
\[
  p_j = \Lambda(\alpha + \gamma_1 \Delta \mu_j + \gamma_2 \Delta R_j).
\]
Throughout this design discussion temperature is held fixed and observed, so the information problem is governed by menu placement and query allocation rather than by rescaling the latent index.

Suppose the sampled-choice logit index is correctly specified and the design matrix with rows $w_j=(1,x_j')'$ has full rank. Then choice data identify the intercept and the reduced-form slopes $(\gamma_1,\gamma_2) = (\kappa/\tau,\,-\kappa\beta/\tau)$. If $\gamma_1 \neq 0$, the same data identify $\beta = -\gamma_2/\gamma_1$. If $\tau$ is known, then $\kappa = \tau\gamma_1$ is identified. If $\tau$ is not observed, then $\kappa$ and $\tau$ are not separately identified.

When private models are queried repeatedly on the same menu, the Bernoulli likelihood from Equation~\eqref{eq:log-likelihood} is equivalently written in grouped binomial form. If menu $j$ has $n_j$ valid parsed A/B responses and option A is chosen $y_j$ times, then
\[
  \mathcal{L}_{\mathrm{grp}}(\alpha,\gamma_1,\gamma_2)
  = \sum_{j=1}^J
  \left[
    y_j \ln p_j
    + (n_j-y_j)\ln(1-p_j)
  \right].
\]
This grouped likelihood uses the sufficient statistics for repeated draws and is numerically equivalent to the trial-level Bernoulli likelihood under conditional independence. Invalid or ambiguous parses are excluded from both $y_j$ and $n_j$ and enter the completeness diagnostic separately.

Because frontier choice probabilities must be formed from finitely many queries rather than read off a logit vector, the design specifies a replication count for each menu class, and these counts differ by provider and by experiment. The schedule is not summarized by a single repetition number, so Table~\ref{tab:frontier-schedule} reports it in full. Two regularities organize it. First, the structural estimation menus and the sampling-based diagnostics are replicated heavily, while the provider that exposes first-token logprobs (Google) requires only a single query for the diagnostics whose indices can be read directly from the recorded probabilities. The Anthropic and OpenAI APIs expose no logprobs, so their diagnostics are identified from repeated sampling alone. Second, the higher replication counts are kept affordable by capping the number of menus drawn from the larger diagnostic pools. These caps are stated in the table and reduce coverage relative to the open-weight counts in Table~\ref{tab:diagnostic-battery} rather than expanding it. All frontier runs fix $\tau = 1$ and top-$p = 1$.

\begin{table}[htbp]
  \centering
  \small
  \resizebox{\textwidth}{!}{
    \begin{tabular}{lcccc}
      \toprule
      Menu class & Anthropic (Claude 4.x) & OpenAI (GPT-4x) & Google (Gemini 2.5) & Menu cap \\
      \midrule
      Estimation & 20 & 20 & 20 (mirror 20) & 400 menus \\
      Placebo (reflexivity) & 20 & 20 & 1 (logprob) & 400 menus \\
      Monotonicity & 20 & 20 & 1 (logprob) & 400 menus \\
      Transitivity (triad legs) & 25 & 25 & 1 (logprob) & 400 triads \\
      Independence (trinary) & 50 & 50 & 1 (logprob) & 400 base menus \\
      Continuity (8 path points) & 25 & 25 & 1 (logprob) & 400 menus \\
      \bottomrule
    \end{tabular}
  }
  \caption{Frontier replication schedule (queries per menu) by provider and menu class.}
  \label{tab:frontier-schedule}
\end{table}

Design is now local because the expected Fisher information depends on the unknown choice probabilities. For binary logit with one draw from each of $N$ distinct menus, $p_j = \Lambda(\alpha + \gamma_1 \Delta \mu_j + \gamma_2 \Delta R_j)$, and $w_j=(1,x_j')'$, the expected Fisher information for $(\alpha,\gamma_1,\gamma_2)$ is
\[
  \frac{1}{N}\sum_{j=1}^N p_j(1-p_j)w_jw_j'.
\]
The corresponding slope-information block is
\[
  M_N^{\mathrm{RP}} = \frac{1}{N}\sum_{j=1}^N p_j(1-p_j)x_jx_j'.
\]
If menu $j$ is instead queried $n_j$ times, the expected information becomes
\[
  I(\mathbf{n})
  =
  \sum_{j=1}^J n_j p_j(1-p_j)w_jw_j'.
\]
The weight $p_j(1-p_j)$ is largest near indifference, so menus with extreme predicted choice probabilities contribute little information even when $x_j$ is large. Once pilot estimates are available, efficient sampled-choice design therefore operates on both margins of the problem. It chooses where to place menus and how to allocate repeated queries across those menus.

For a fixed total query budget $T = \sum_{j=1}^J n_j$, write $q_j = n_j/T$. Following \cite{louviere1983design} and \cite{louviere2010discrete}, a locally D-optimal second wave solves
\[
  q^{D}
  \in
  \arg\max_{\{q_j \ge 0,\ \sum_{j=1}^J q_j = 1\}}
  \det\left(
    \sum_{j=1}^J q_j \tilde{p}_j(1-\tilde{p}_j)w_jw_j'
  \right),
\]
where $\tilde{p}_j = \Lambda(\tilde{\alpha} + \tilde{\gamma}_1 \Delta \mu_j + \tilde{\gamma}_2 \Delta R_j)$ is evaluated at pilot estimates. Since the gradient of $\beta$ with respect to $(\alpha,\gamma_1,\gamma_2)$ is proportional to $(0,\beta,1)$, a locally c-optimal second wave for the ratio parameter solves
\begin{equation}\label{eq:sampled-c-opt-design}
  q^{\beta}
  \in
  \arg\min_{\{q_j \ge 0,\ \sum_{j=1}^J q_j = 1\}}
  (0,\,\tilde{\beta},\,1)
  \left(
    \sum_{j=1}^J q_j \tilde{p}_j(1-\tilde{p}_j)w_jw_j'
  \right)^{-1}
  (0,\,\tilde{\beta},\,1)'.
\end{equation}
These criteria imply a two-stage design. A first wave should be broad, balanced, and approximately orthogonal so the signs and relative magnitudes of the slopes can be learned. A second wave should concentrate on menus near the estimated indifference frontier $\alpha + \gamma_1 \Delta \mu_j + \gamma_2 \Delta R_j \approx 0$ while preserving independent movement in $\Delta \mu_j$ and $\Delta R_j$. It should also place more repetitions on those menus rather than spreading the query budget uniformly across the full first wave.

Under the common first-wave menu set $\mathcal{J}$ used in this paper, equal replication $n_j = m$ for every $j \in \mathcal{J}$ gives
\[
  I_m
  =
  m \sum_{j\in\mathcal{J}} p_j(1-p_j)w_jw_j'.
\]
Applying the Delta method to $\beta = -\gamma_2/\gamma_1$ then yields
\[
  \operatorname{Avar}(\hat{\beta}\mid \mathcal{J},m)
  \approx
  \frac{\tau^2}{m\kappa^2}
  (0,\,\beta,\,1)
  \left(
    \sum_{j\in\mathcal{J}} p_j(1-p_j)w_jw_j'
  \right)^{-1}
  (0,\,\beta,\,1)'.
\]
Consequently, a target standard error $s_{\beta}$ is met whenever $m$ satisfies
\begin{equation}\label{eq:min-reps}
  m
  \ge
  \frac{\tau^2}{\kappa^2 s_{\beta}^2}
  (0,\,\beta,\,1)
  \left(
    \sum_{j\in\mathcal{J}} p_j(1-p_j)w_jw_j'
  \right)^{-1}
  (0,\,\beta,\,1)'.
\end{equation}
Equation~\eqref{eq:min-reps} therefore provides a precision benchmark for equal replication. One can evaluate it over a conservative grid of plausible $(\beta,\kappa)$ values and take the largest implied integer. The main experiment in this paper uses equal replication across the common first-wave menus because it preserves comparability across observation regimes and model versions. This is transparent but not generally optimal for sampled-choice estimation. A model-specific second wave that reweights the first-wave menus according to Equation~\eqref{eq:sampled-c-opt-design}, or augments them with new menus near the pilot indifference frontier, is more statistically efficient when cross-model comparability is not the primary objective. Varying temperature across menus is not a substitute for this design because it only rescales the same index. As a robustness exercise, one can repeat the estimator at alternative fixed temperatures and verify that $\hat{\beta}$ remains stable while $\hat{\kappa}$ rescales according to $\hat{\kappa} = \tau \hat{\gamma}_1$.

\section{Model Fine-tuning Methodology}\label{app:fine-tuning}

This appendix specifies the construction of the two fine-tuned validation models reported in Section~\ref{subsec:fine-tuning}. It states the random-utility foundation of the training target, the full multi-term objective, the training corpus, the selection of the logit scale, the comparison with the earlier structural-gap objective, and the functional forms of the diagnostics used to evaluate the models.

For orientation, the fine-tuning exercise has two moving parts. First, the structural term teaches the model a logit surface, where for each portfolio menu, the A-minus-B label-logit gap should equal $\kappa$ times the utility difference implied by the target $\beta$. Second, the auxiliary invariance terms prevent degenerate ways of achieving that target, such as always favoring the first label, assigning mass outside the answer set, or changing the A/B odds when a dominated third option is added. The training target is therefore a relation between portfolio attributes and logit gaps, not a fixed preference for any answer token.

\subsection{Targeting Logits to a Parameterized Utility}\label{app:ft-rum}

The fine-tuned models are trained to embody the maintained utility $V_i = \mu_i - \beta R_i$ for a prescribed value of $\beta$. This training objective rests on a random-utility interpretation of the model's single-token choice. Let the latent value of option $i$ in a menu be $u_i = V_i + \varepsilon_i$, where $V_i = \mu_i - \beta R_i$ is the maintained benchmark utility and $\varepsilon_i$ is an idiosyncratic disturbance. When the disturbances are independent and identically Gumbel distributed, the probability that option $A$ is chosen over option $B$ is the binary logit
\begin{equation}\label{eq:rum-logit}
  P(A \succ B) = \frac{\exp(\frac{\kappa}{\tau} V_A)}{\exp(\frac{\kappa}{\tau} V_A) + \exp(\frac{\kappa}{\tau} V_B)},
\end{equation}
where $\kappa$ is the inverse scale of the Gumbel disturbance and $\tau$ is the model's temperature parameter. The softmax that the model places over the two label tokens is the empirical analog of this choice probability, so the pre-softmax logit gap between the two labels is the empirical analog of $\kappa(V_A - V_B)$. The structural-gap loss $\mathcal{L}_{\text{struct}}$ trains the model to match its position-corrected logit gap to this target over the jittered-grid estimation corpus, after removing the additive positional component identified on placebo menus.

The scale $\kappa$ enters Equation~\eqref{eq:rum-logit} only as a common multiplier of the utility difference and is not separately identified from the strength of preference within a single binary menu. The recovered risk parameter is unaffected by this indeterminacy. Writing the restricted quadratic surface fit to the position-corrected indifference data as $\Delta u = \theta_1 \Delta\mu + \theta_2 \Delta R$, the estimator of risk aversion is the ratio
\begin{equation}\label{eq:beta-ratio}
  \hat{\beta} = -\frac{\hat{\theta}_2}{\hat{\theta}_1}.
\end{equation}
Multiplying the entire logit surface by any positive constant $c$ rescales both $\hat{\theta}_1$ and $\hat{\theta}_2$ by $c$ and leaves their ratio unchanged. The choice of $\kappa$ therefore affects only the numerical conditioning of the training problem and the precision with which the indifference locus is identified, not the recovered value of $\beta$. This is the formal basis for the invariance asserted in Section~\ref{subsec:fine-tuning}. Under the original structural-gap objective the loss is computed with $\kappa$ set equal to the observed logit-gap standard deviation of the base model, a convention we denote \texttt{kappa\_anchor = true}. Section~\ref{app:ft-kappa} explains why this convention fails at low $\beta$ and how we replace it.

\subsection{The Multi-Term Invariance Objective}

All loss terms are evaluated at the cue position, immediately before the model emits its one-token response. Let $z_m\in\mathbb{R}^{|\mathcal V|}$ denote these logits for menu $m$. Since an answer label can have more than one tokenization, let $\mathcal T_q\subset\mathcal V$ be the set of token ids decoding to label $q\in\{A,B,C\}$ and define
\begin{equation}\label{eq:grouped-label-logits}
  \ell_q(z)=\log\sum_{v\in\mathcal T_q}\exp(z_v),
  \qquad
  \pi_v(z)=\frac{\exp(z_v)}{\sum_{w\in\mathcal V}\exp(z_w)}.
\end{equation}
For a binary menu $m=(A,B)$, write $g_m=\ell_A(z_m)-\ell_B(z_m)$. The training utility is
\begin{equation}\label{eq:training-utility}
  U_\beta(\mu,\sigma)=\mu-\beta(\mu^2+\sigma^2),
  \qquad
  \Delta U_m=U_\beta(\mu_A,\sigma_A)-U_\beta(\mu_B,\sigma_B).
\end{equation}
All expectations below are empirical averages over the indicated training set. For scalar residual $r$, let
\begin{equation}\label{eq:huber-loss}
  \rho_\delta(r)=
  \begin{cases}
    \frac{1}{2}r^2, & |r|\leq \delta,\\
    \delta(|r|-\frac{1}{2}\delta), & |r|>\delta,
  \end{cases}
\end{equation}
with $\delta=1$ throughout.

At optimizer step $t$, the full objective is
\begin{align}\label{eq:inv-loss}
  \mathcal L_t
  &=\mathcal L_{\mathrm{struct}}
  +\lambda_{\mathrm{mass}}\mathcal L_{\mathrm{mass}}
  +\lambda_{\mathrm{sob}}\mathcal L_{\mathrm{sob}} \nonumber\\
  &\quad
  +a_t\Big[
    \lambda_{\mathrm{ref}}\mathcal L_{\mathrm{ref}}
    +\lambda_{\mathrm{mir}}\mathcal L_{\mathrm{mir}}
    +\lambda_{\mathrm{iia}}\{\mathcal L_{\mathrm{iia}}
    +\lambda_{\mathrm{disp}}\mathcal L_{\mathrm{iia,disp}}\} \nonumber\\
    &\qquad\qquad
    +\lambda_{\mathrm{exact}}\{\mathcal L_{\mathrm{exact}}
    +\lambda_{\mathrm{disp}}\mathcal L_{\mathrm{exact,disp}}\}
    +\lambda_{3\mathrm{m}}\mathcal L_{3\mathrm{m}}
    +\lambda_{\mathrm{cyc}}\mathcal L_{\mathrm{cyc}}
    +\lambda_{\mathrm{path}}\mathcal L_{\mathrm{path}}
  \Big],
\end{align}
where $a_t=\min\{1,(t+1)/T_{\mathrm{aux}}\}$ and $T_{\mathrm{aux}}$ is ten percent of the planned optimizer steps. The structural, binary-mass, and Sobolev terms enter from the first step; the remaining terms are linearly introduced through $a_t$. In the reported full-objective runs,
\[
  (\lambda_{\mathrm{mass}},\lambda_{\mathrm{sob}},\lambda_{\mathrm{ref}},
    \lambda_{\mathrm{mir}},\lambda_{\mathrm{iia}},\lambda_{\mathrm{disp}},
    \lambda_{\mathrm{exact}},\lambda_{3\mathrm{m}},\lambda_{\mathrm{cyc}},
  \lambda_{\mathrm{path}},\lambda_{\mathrm{path,smooth}})
\]
\[
  =(0.10,0.10,0.10,0.20,0.05,0.05,0.15,0.05,0.10,0.15,0.03).
\]
Terms with zero weight, or without corresponding training examples in a run, are absent from that run.

The structural term matches the binary logit gap to the random-utility target:
\begin{equation}\label{eq:loss-struct}
  \mathcal L_{\mathrm{struct}}
  =
  \mathbb E_{m\in\mathcal D_2}
  \rho_\delta\!\left(g_m-b_\alpha-\kappa\Delta U_m\right).
\end{equation}
Here $\mathcal D_2$ denotes the binary structural training rows. The invariance runs set $b_\alpha=0$; the earlier structural-gap runs set it equal to the pre-training estimate of the positional offset.

The binary mass term is
\begin{equation}\label{eq:loss-mass}
  \mathcal L_{\mathrm{mass}}
  =
  \mathbb E_{m\in\mathcal D_2}
  \left[-\log\sum_{v\in\mathcal T_A\cup\mathcal T_B}\pi_v(z_m)\right],
\end{equation}
which penalizes probability assigned outside the admissible answer labels.

The Sobolev term constrains local derivatives of the logit surface. For derivative rows $m\in\mathcal D_\partial$, let $z_{m,\mu+}$ and $z_{m,\mu-}$ be the cue logits after perturbing $\mu_A$ by $\pm h_\mu$, and define $z_{m,\sigma+}$ and $z_{m,\sigma-}$ analogously. Set
\[
  D_\mu g_m=\frac{g(z_{m,\mu+})-g(z_{m,\mu-})}{2h_\mu},
  \qquad
  D_\sigma g_m=\frac{g(z_{m,\sigma+})-g(z_{m,\sigma-})}{2h_\sigma},
\]
where $g(z)=\ell_A(z)-\ell_B(z)$. Since
\[
  \frac{\partial\Delta U_m}{\partial\mu_A}=1-2\beta\mu_A,
  \qquad
  \frac{\partial\Delta U_m}{\partial\sigma_A}=-2\beta\sigma_A,
\]
the implemented loss is
\begin{equation}\label{eq:loss-sobolev}
  \mathcal L_{\mathrm{sob}}
  =
  \mathbb E_{m\in\mathcal D_\partial}
  \frac{1}{2}\left[
    \rho_\delta\!\left(D_\mu g_m-\kappa(1-2\beta\mu_A)\right)
    +\rho_\delta\!\left(D_\sigma g_m-\kappa(-2\beta\sigma_A)\right)
  \right].
\end{equation}

The reflexivity and mirror terms impose label-position invariance. For placebo menus $\mathcal D_0$ in which both labels describe the same portfolio,
\begin{equation}\label{eq:loss-reflexive}
  \mathcal L_{\mathrm{ref}}
  =
  \mathbb E_{m\in\mathcal D_0}\rho_\delta(g_m).
\end{equation}
For each binary menu $m$ and its label-swapped mirror $\tilde m$,
\begin{equation}\label{eq:loss-mirror}
  \mathcal L_{\mathrm{mir}}
  =
  \mathbb E_{(m,\tilde m)}\rho_\delta(g_m+g_{\tilde m}).
\end{equation}

The dominated-alternative terms are defined on trinary examples $e=(A,B,C)$, where $C$ is strictly dominated by both original portfolios. Let $\Pi$ denote the six permutations of the three portfolios across the answer slots. If $s_A(\pi)$ and $s_B(\pi)$ are the slots containing the original portfolios $A$ and $B$ under permutation $\pi$, define
\[
  g_{e\pi}^{(3)}
  =
  \ell_{s_A(\pi)}(z_{e\pi})-\ell_{s_B(\pi)}(z_{e\pi}),
  \qquad
  \bar g_e^{(3)}=\frac{1}{|\Pi|}\sum_{\pi\in\Pi}g_{e\pi}^{(3)}.
\]
The value and dispersion components are
\begin{align}
  \mathcal L_{\mathrm{iia}}
  &=
  \mathbb E_e
  \rho_\delta\!\left(\bar g_e^{(3)}-\kappa\Delta U_{AB,e}\right),
  \label{eq:loss-iia}\\
  \mathcal L_{\mathrm{iia,disp}}
  &=
  \mathbb E_e\frac{1}{|\Pi|}\sum_{\pi\in\Pi}
  \rho_\delta\!\left(g_{e\pi}^{(3)}-\kappa\Delta U_{AB,e}\right).
  \label{eq:loss-iia-disp}
\end{align}
Thus the average trinary $A/B$ gap is tied to the structural target, while the dispersion term rules out offsetting slot-specific deviations.

The exact IIA term instead anchors the trinary odds to the model's own binary odds for the same pair. Let $g_e^{(2)}$ denote the binary gap and let $\mathrm{sg}(\cdot)$ denote stop-gradient evaluation. Then
\begin{align}
  \mathcal L_{\mathrm{exact}}
  &=
  \mathbb E_e
  \rho_\delta\!\left(\bar g_e^{(3)}-\mathrm{sg}(g_e^{(2)})\right),
  \label{eq:loss-exact-iia}\\
  \mathcal L_{\mathrm{exact,disp}}
  &=
  \mathbb E_e\frac{1}{|\Pi|}\sum_{\pi\in\Pi}
  \rho_\delta\!\left(g_{e\pi}^{(3)}-\mathrm{sg}(\bar g_e^{(3)})\right).
  \label{eq:loss-exact-iia-disp}
\end{align}
The associated trinary mass term is
\begin{equation}\label{eq:loss-trinary-mass}
  \mathcal L_{3\mathrm{m}}
  =
  \mathbb E_{e,\pi}
  \left[-\log\sum_{v\in\mathcal T_A\cup\mathcal T_B\cup\mathcal T_C}
  \pi_v(z_{e\pi})\right].
\end{equation}

The cycle term is defined on binary triads $(A,B)$, $(B,C)$, and $(A,C)$. Let $g_{e,AB}$, $g_{e,BC}$, and $g_{e,AC}$ be the corresponding logit gaps, and let $s_{e,r}\in\{-1,0,1\}$ be the target sign for comparison $r\in\{AB,BC,AC\}$. With $x_+=\max\{x,0\}$ and $\mathcal J_e=\{r:s_{e,r}\neq0\}$,
\begin{equation}\label{eq:loss-cycle}
  \mathcal L_{\mathrm{cyc}}
  =
  \mathbb E_e\left[
    \rho_\delta(g_{e,AB}+g_{e,BC}-g_{e,AC})
    +\omega_{\mathrm{sgn}}\frac{1}{|\mathcal J_e|}
    \sum_{r\in\mathcal J_e}(m_{\mathrm{sgn}}-s_{e,r}g_{e,r})_+
  \right],
\end{equation}
where $\omega_{\mathrm{sgn}}=0.25$ and $m_{\mathrm{sgn}}=0.25$.

Finally, the path term is defined on interpolation paths $p$ with ordered gaps $g_{p1},\ldots,g_{pK}$:
\begin{align}
  \mathcal L_{\mathrm{path}}
  &=
  \mathbb E_p\left[
    (m_{\mathrm{br}}+\min_k g_{pk})_+
    +(m_{\mathrm{br}}-\max_k g_{pk})_+ \right. \nonumber\\
    &\qquad\qquad\left.
    +\lambda_{\mathrm{path,smooth}}
    \frac{1}{K-2}\sum_{k=2}^{K-1}
    (g_{p,k+1}-2g_{pk}+g_{p,k-1})^2
  \right],
  \label{eq:loss-path}
\end{align}
with $m_{\mathrm{br}}=0.05$. The first two terms require the path to bracket the indifference surface; the last term penalizes discrete curvature along the path.

\subsection{Training Corpus}

For the high-$\beta$ run the corpus follows the original structural-gap construction. It contains 2,048 jittered-grid estimation menus, 2,048 dominance menus, 2,048 iso-utility menus, 600 derivative triplets for the Sobolev term, 1,024 IIA dominated menus, and 30 edge menus at each of $\Delta\mu \approx 0$ and $\Delta R \approx 0$. The grid spans $\mu \in [2, 26]$ in steps of $0.5$ and $\sigma \in [10, 80]$ in steps of $2.5$, and each menu is paired with its positional mirror.

For the low-$\beta$ runs the corpus is enlarged in two respects. The iso-utility count is increased from 2,048 to 4,096 to provide denser coverage of the nearly flat indifference contours that are harder to pin down at very low $\beta$. The edge coverage is expanded from 30 to 400 menus per edge type to provide more curvature-isolating examples at the boundary of the identification region. Training runs for five epochs rather than the two used for the high-$\beta$ case, allowing the auxiliary losses to propagate through the full parameter trajectory.

\subsection{Logit Scale and the Low-\texorpdfstring{$\beta$}{beta} Kappa Problem}\label{app:ft-kappa}

The logit scale $\kappa$ is the central design choice for the low-$\beta$ training. As Section~\ref{app:ft-rum} noted, the structural-gap objective by default anchors $\kappa$ to the observed logit-gap standard deviation of the base model. At $\beta = 0.0045$ the base model's logit scale is moderate and anchoring works well. At $\beta = 0.0001$ the base model's gap standard deviation is small enough that the structured component of the logit is numerically negligible in bf16 arithmetic. Anchoring $\kappa$ to this small value forces the training signal into the noise floor and leaves positional structure to dominate, which is the failure mode that the reflexivity terms are meant to cure. The remedy is to set $\kappa$ to a fixed positive constant that places the logit gaps at a numerically stable scale. Because $\hat{\beta}$ is invariant to $\kappa$ by the argument of Equation~\eqref{eq:beta-ratio}, this choice does not bias the structural estimate, and it affects only numerical stability and the shape of the learned preference surface. To select the constant we trained three otherwise identical models at $\beta = 0.0001$ with $\kappa \in \{1, 2, 4\}$. Table~\ref{tab:kappa-sweep} reports the resulting diagnostics.

\begin{table}[htbp]
  \centering
  \small
  \begin{tabular}{lccccccc}
    \toprule
    $\kappa$ target & $\mathcal{C}$ & $\mathcal{R}$ &
    $\mathcal{K}_G$ & $\mathcal{M}$ & $\mathcal{T}$ & $\mathcal{I}$ &
    $\hat\kappa$ \\
    \midrule
    1.0 & 0.9961 & 0.9941 & 0.9333 & 1.0000 & 1.0000 & 0.9484 & 1.001 \\
    2.0 & 0.9942 & 0.9945 & 0.7398 & 1.0000 & 1.0000 & 0.8569 & 2.003 \\
    4.0 & 0.9904 & 0.9959 & 0.7825 & 1.0000 & 1.0000 & 0.4960 & 3.991 \\
    \bottomrule
  \end{tabular}
  \caption{Diagnostic indices across the $\kappa$ sweep at $\beta = 0.0001$.}
  \label{tab:kappa-sweep}
\end{table}

IIA deteriorates sharply as $\kappa$ increases, with the mean index falling from $0.948$ at $\kappa = 1$ to $0.497$ at $\kappa = 4$. A larger $\kappa$ magnifies the utility gap between options and makes the model more sensitive to the precise content of each option, so that appending a dominated third option induces a re-evaluation of the relative standing of $A$ and $B$ that violates IIA. Continuity is also highest at $\kappa = 1$. The $\kappa = 1$ run additionally delivers the cleanest structural recovery, returning $\hat\kappa = 1.001$ and $\hat\beta = 8.68 \times 10^{-5} \approx 0.0001$. We therefore select $\kappa = 1$ for the low-$\beta$ model.

\subsection{Selected Models and Comparison with the Structural-Gap Objective}

Two models are carried into Section~\ref{subsec:fine-tuning}. The high-$\beta$ model is trained on the high-$\beta$ corpus with anchoring retained, a nominal $\kappa = 0.20$, and two epochs in bf16, and it converged in 0.89 hours. The low-$\beta$ model is trained on the enlarged low-$\beta$ corpus with $\kappa = 1.0$ and anchoring disabled over five epochs in bf16, and it converged in 2.49 hours with a recovered in-batch scale of $1.000$. Table~\ref{tab:inv-results} compares the two selected models with the earlier structural-gap models on the full set of diagnostics and structural estimates.

\begin{table}[htbp]
  \centering
  \small
  \setlength{\tabcolsep}{4pt}
  \begin{tabular}{lccccccccc}
    \toprule
    Model & $\mathcal{C}$ & $\mathcal{R}$ & $\mathcal{K}_G$ &
    $\mathcal{M}$ & $\mathcal{T}$ & $\mathcal{I}$ &
    $\hat\kappa$ & $\hat\beta$ & $R^2$ \\
    \midrule
    \multicolumn{10}{l}{\textit{Structural-gap objective}} \\
    \midrule
    $\beta=0.0045$ &
    0.9977 & 0.9966 & 1.0000 & 1.0000 & 0.9994 & 0.5870 &
    0.6597 & 0.0045 & 0.9979 \\
    $\beta=0.0001$ &
    0.9971 & 0.6057 & 0.9667 & 1.0000 & 1.0000 & 0.3554 &
    0.7333 & 0.0001 & 0.9997 \\
    \midrule
    \multicolumn{10}{l}{\textit{Invariance objective}} \\
    \midrule
    $\beta=0.0045$ &
    0.9914 & 0.9870 & 1.0000 & 1.0000 & 0.9583 & 0.8000 &
    0.7189 & 0.0046 & 0.9935 \\
    $\beta=0.0001$, $\kappa=1$ &
    0.9961 & 0.9941 & 0.9333 & 1.0000 & 1.0000 & 0.9484 &
    1.0013 & 0.0001 & 0.9992 \\
    \bottomrule
  \end{tabular}
  \caption{Diagnostic indices and structural estimates for the structural-gap and invariance objectives at each $\beta$ target.}
  \label{tab:inv-results}
\end{table}

For the high-$\beta$ model the principal gain is in IIA, where the index rises from $0.587$ to $0.800$, while reflexivity remains near unity and monotonicity is maintained at one. Transitivity falls modestly from $0.999$ to $0.958$, so that five of the one hundred and twenty directional triads now cycle, and continuity remains at one. The recovered parameters are close to the training target, with $\hat\beta = 0.004560$ against a target of $0.0045$ and $R^2 = 0.9935$. For the low-$\beta$ model the improvement is larger. Reflexivity rises from $0.606$ to $0.994$, with the first-option share falling from $0.573$ to $0.499$ and the intercept-only position bias $\hat\alpha = -0.004$ (standard error $0.002$) statistically indistinguishable from zero. IIA rises from $0.355$ to $0.948$, the strongest value in the subject pool, and monotonicity is maintained at one. Continuity falls slightly, from $0.967$ to $0.933$, reflecting two interpolation paths whose log-odds are monotone but fail to cross zero.\footnote{A run in fp32 arithmetic was initiated to test whether this continuity reduction is an artifact of bf16 quantization. The run was not carried to completion, so whether fp32 arithmetic would recover the two non-bracketing paths remains an open question.} The structural recovery is essentially exact, with $\hat\kappa = 1.001$, $\hat\beta = 8.68 \times 10^{-5}$, and $R^2 = 0.9992$.

\subsection{Reproducibility}

The resolved training configurations, the LoRA adapters \citep{hu2022lora} and merged weights for each selected model, and the per-axiom evaluation outputs are provided in the replication package that accompanies the paper. The high-$\beta$ and low-$\beta$ runs share the loss architecture of Equation~\eqref{eq:inv-loss} and differ only in the corpus enlargements, the epoch count, and the logit-scale convention documented above.

\section{Robustness Checks}\label{app:robustness}

\subsection{Prompt Variations}

To ensure our results are not merely artifacts of the specific phrasing used in our main experimental prompt, we conducted robustness checks using two alternative prompt structures. The underlying design grid of choice menus remains completely unchanged; only the textual presentation of the problem differs.

Below, we document the three prompts tested. The \textbf{Baseline Prompt} is the phrasing used throughout the main body of the paper. The \textbf{Compact Table Prompt} removes most conversational framing and presents the options in a dense, structured format. The \textbf{Sentence Neutral Prompt} rephrases the task instruction to use neutral synonyms (e.g., "comparing" instead of "options", "select" instead of "choose").

\subsubsection*{Baseline Prompt}
\begin{tcolorbox}[colback=gray!5,colframe=gray!50,arc=4mm,boxrule=1pt]
  \ttfamily
  You have two options for your annual financial portfolio.\\
  Choose the portfolio you prefer and reply with A or B only.

  Option A offers an expected return of \detokenize{{mu_A}}\% and a standard deviation\\
  of \detokenize{{sigma_A}}\%.

  Option B offers an expected return of \detokenize{{mu_B}}\% and a standard deviation\\
  of \detokenize{{sigma_B}}\%.
\end{tcolorbox}

\subsubsection*{Compact Table Prompt}
\begin{tcolorbox}[colback=gray!5,colframe=gray!50,arc=4mm,boxrule=1pt]
  \ttfamily
  Choose one annual financial portfolio. Reply with only A or B.

  A: expected return = \detokenize{{mu_A}}\%; standard deviation = \detokenize{{sigma_A}}\%.\\
  B: expected return = \detokenize{{mu_B}}\%; standard deviation = \detokenize{{sigma_B}}\%.
\end{tcolorbox}

\subsubsection*{Sentence Neutral Prompt}
\begin{tcolorbox}[colback=gray!5,colframe=gray!50,arc=4mm,boxrule=1pt]
  \ttfamily
  You are comparing two annual financial portfolios.\\
  Select the portfolio you prefer. Your answer must be exactly A or B.

  Portfolio A has an expected return of \detokenize{{mu_A}}\% and a standard deviation of \detokenize{{sigma_A}}\%.\\
  Portfolio B has an expected return of \detokenize{{mu_B}}\% and a standard deviation of \detokenize{{sigma_B}}\%.
\end{tcolorbox}

Table~\ref{tab:rationality-diagnostics-robustness} and Table~\ref{tab:structural-estimates-robustness} present the consolidated rationality diagnostics and structural parameter estimates, respectively, across these three prompt variations for all open-weight models. Table~\ref{tab:menu-design-robustness} reports the corresponding structural robustness check for random menu creation from the admissible portfolio-pair grid.

\begin{table}[htbp]
  \centering
  \small
  \resizebox{\textwidth}{!}{
    \begin{tabular}{llcccccc}
      \toprule
      Model & Prompt & Complete & Reflexive & Continuity & Monotonicity & Transitive & IIA \\
      \midrule
      \multirow{3}{*}{Qwen 3 (14B)} & Baseline & 1.0000 & 0.0009 & 1.0000 & 1.0000 & 0.9831 & 0.3215 \\
      & Compact Table & 1.0000 & 0.0020 & 1.0000 & 1.0000 & 0.9655 & 0.4489 \\
      & Sentence Neutral & 1.0000 & 0.0001 & 1.0000 & 1.0000 & 1.0000 & 0.4804 \\
      \midrule
      \multirow{3}{*}{Qwen 3 (8B)} & Baseline & 1.0000 & 0.0144 & 1.0000 & 1.0000 & 0.9565 & 0.6437 \\
      & Compact Table & 1.0000 & 0.0286 & 1.0000 & 1.0000 & 0.9402 & 0.5764 \\
      & Sentence Neutral & 1.0000 & 0.1245 & 1.0000 & 1.0000 & 0.9008 & 0.6218 \\
      \midrule
      \multirow{3}{*}{Gemma 4 (31B Instruct)} & Baseline & 1.0000 & 0.0000 & 1.0000 & 1.0000 & 0.9680 & 0.1222 \\
      & Compact Table & 1.0000 & 0.0000 & 1.0000 & 1.0000 & 0.9835 & 0.2813 \\
      & Sentence Neutral & 1.0000 & 0.0000 & 1.0000 & 1.0000 & 0.9920 & 0.2599 \\
      \midrule
      \multirow{3}{*}{Ministral 3 (8B Instruct)} & Baseline & 0.9995 & 0.0876 & 1.0000 & 1.0000 & 0.9937 & 0.8242 \\
      & Compact Table & 0.9982 & 0.7246 & 1.0000 & 1.0000 & 0.8926 & 0.8601 \\
      & Sentence Neutral & 0.9892 & 0.2944 & 1.0000 & 1.0000 & 0.9426 & 0.8417 \\
      \midrule
      \multirow{3}{*}{Gemma 2 (9B Instruct)} & Baseline & 0.9960 & 0.7974 & 1.0000 & 1.0000 & 0.9855 & 0.3272 \\
      & Compact Table & 0.9982 & 0.4866 & 1.0000 & 1.0000 & 0.9508 & 0.3099 \\
      & Sentence Neutral & 0.9982 & 0.1730 & 1.0000 & 1.0000 & 0.8899 & 0.3485 \\
      \midrule
      \multirow{3}{*}{Llama 3.1 (8B Instruct)} & Baseline & 0.9021 & 0.8364 & 1.0000 & 1.0000 & 0.9938 & 0.9195 \\
      & Compact Table & 0.9640 & 0.8087 & 0.9750 & 1.0000 & 1.0000 & 0.8990 \\
      & Sentence Neutral & 0.9245 & 0.4003 & 1.0000 & 1.0000 & 0.9793 & 0.8409 \\
      \midrule
      \multirow{3}{*}{Llama 3.1 (70B Instruct)} & Baseline & 0.7580 & 0.2689 & 1.0000 & 1.0000 & 1.0000 & 0.7627 \\
      & Compact Table & 0.7628 & 0.3252 & 1.0000 & 1.0000 & 0.9626 & 0.8339 \\
      & Sentence Neutral & 0.6699 & 0.1124 & 1.0000 & 1.0000 & 0.9826 & 0.7990 \\
      \bottomrule
    \end{tabular}
  }
  \caption{Rationality diagnostics across prompt variations}
  \label{tab:rationality-diagnostics-robustness}
\end{table}

\begin{table}[htbp]
  \centering
  \small
  \begin{tabular}{llccc}
    \toprule
    Model & Prompt & $\hat{\kappa}$ & $\hat{\beta}$ & $R^2$ \\
    \midrule
    \multirow{6}{*}{Qwen 3 (14B)} & Baseline & +0.631835 & +0.003302 & 0.9360 \\
    & & (0.006339) & (0.0000) & \\
    & Compact Table & +0.506505 & +0.003936 & 0.9187 \\
    & & (0.004328) & (0.0001) & \\
    & Sentence Neutral & +0.597744 & +0.002097 & 0.9519 \\
    & & (0.005575) & (0.0000) & \\
    \midrule
    \multirow{6}{*}{Qwen 3 (8B)} & Baseline & +0.310677 & +0.003057 & 0.8887 \\
    & & (0.004192) & (0.0001) & \\
    & Compact Table & +0.325675 & +0.003343 & 0.8906 \\
    & & (0.003444) & (0.0001) & \\
    & Sentence Neutral & +0.186353 & +0.003589 & 0.9189 \\
    & & (0.002105) & (0.0000) & \\
    \midrule
    \multirow{6}{*}{Gemma 4 (31B Instruct)} & Baseline & +1.482339 & +0.004464 & 0.9771 \\
    & & (0.007633) & (0.0000) & \\
    & Compact Table & +1.160998 & +0.004422 & 0.9487 \\
    & & (0.008343) & (0.0001) & \\
    & Sentence Neutral & +1.467195 & +0.004985 & 0.9760 \\
    & & (0.008035) & (0.0000) & \\
    \midrule
    \multirow{6}{*}{Ministral 3 (8B Instruct)} & Baseline & +0.097941 & +0.011363 & 0.8968 \\
    & & (0.003656) & (0.0003) & \\
    & Compact Table & +0.115390 & +0.004348 & 0.9426 \\
    & & (0.001158) & (0.0000) & \\
    & Sentence Neutral & +0.110512 & +0.004349 & 0.8966 \\
    & & (0.001571) & (0.0001) & \\
    \midrule
    \multirow{6}{*}{Gemma 2 (9B Instruct)} & Baseline & +0.249116 & +0.007871 & 0.9114 \\
    & & (0.005723) & (0.0001) & \\
    & Compact Table & +0.235953 & +0.006484 & 0.9139 \\
    & & (0.004397) & (0.0001) & \\
    & Sentence Neutral & +0.330350 & +0.005079 & 0.8809 \\
    & & (0.004534) & (0.0001) & \\
    \midrule
    \multirow{6}{*}{Llama 3.1 (8B Instruct)} & Baseline & +0.038552 & +0.009121 & 0.8861 \\
    & & (0.002197) & (0.0004) & \\
    & Compact Table & +0.020381 & +0.000152 & 0.8859 \\
    & & (0.000359) & (0.0001) & \\
    & Sentence Neutral & +0.029573 & +0.005770 & 0.7530 \\
    & & (0.000969) & (0.0001) & \\
    \midrule
    \multirow{6}{*}{Llama 3.1 (70B Instruct)} & Baseline & +0.161817 & +0.004846 & 0.9386 \\
    & & (0.001882) & (0.0000) & \\
    & Compact Table & +0.148666 & +0.004984 & 0.9086 \\
    & & (0.001945) & (0.0001) & \\
    & Sentence Neutral & +0.194166 & +0.003228 & 0.9049 \\
    & & (0.002646) & (0.0000) & \\
    \bottomrule
  \end{tabular}
  \caption{Structural estimates across prompt variations}
  \label{tab:structural-estimates-robustness}
\end{table}

\subsection{Menu Design Variations}

Table~\ref{tab:menu-design-robustness} compares the baseline D-optimal menu sample with a random draw from the same admissible non-dominance grid of portfolio pairs. Estimates use the same paired restricted-quadratic estimator as Table~\ref{tab:structural-estimates-robustness}; standard errors are in parentheses. The $N$ column reports retained canonical mirror pairs after parsing and completeness filters.

\begin{table}[htbp]
  \centering
  \small
  \begin{tabular}{llccccc}
    \toprule
    Model & Menu sample & $N$ & $\hat{\alpha}$ & $\hat{\kappa}$ & $\hat{\beta}$ & $R^2$ \\
    \midrule
    \multirow{4}{*}{Qwen 3 (14B)} & D-optimal & 680 & +1.383456 & +0.631835 & +0.003302 & 0.9360 \\
    & & & (0.034971) & (0.006339) & (0.0000) & \\
    & Random grid & 680 & +1.627941 & +0.873192 & +0.002982 & 0.6296 \\
    & & & (0.037589) & (0.022289) & (0.0001) & \\
    \midrule
    \multirow{4}{*}{Qwen 3 (8B)} & D-optimal & 680 & +0.304689 & +0.310677 & +0.003057 & 0.8887 \\
    & & & (0.042850) & (0.004192) & (0.0001) & \\
    & Random grid & 680 & -0.968932 & +0.456527 & +0.003079 & 0.6330 \\
    & & & (0.035984) & (0.011659) & (0.0001) & \\
    \midrule
    \multirow{4}{*}{Gemma 4 (31B Instruct)} & D-optimal & 680 & -0.572503 & +1.482339 & +0.004464 & 0.9771 \\
    & & & (0.037519) & (0.007633) & (0.0000) & \\
    & Random grid & 680 & -0.731503 & +2.026783 & +0.004705 & 0.6905 \\
    & & & (0.059740) & (0.042880) & (0.0001) & \\
    \midrule
    \multirow{4}{*}{Ministral 3 (8B Instruct)} & D-optimal & 680 & +0.383919 & +0.097941 & +0.011363 & 0.8968 \\
    & & & (0.027617) & (0.003656) & (0.0003) & \\
    & Random grid & 680 & +0.360482 & +0.128997 & +0.010003 & 0.4799 \\
    & & & (0.020632) & (0.008756) & (0.0005) & \\
    \midrule
    \multirow{4}{*}{Gemma 2 (9B Instruct)} & D-optimal & 680 & +0.716744 & +0.249116 & +0.007871 & 0.9114 \\
    & & & (0.025175) & (0.005723) & (0.0001) & \\
    & Random grid & 680 & +0.809136 & +0.335482 & +0.007995 & 0.6457 \\
    & & & (0.034776) & (0.012711) & (0.0002) & \\
    \midrule
    \multirow{4}{*}{Llama 3.1 (8B Instruct)} & D-optimal & 200 & +0.172486 & +0.038552 & +0.009121 & 0.8861 \\
    & & & (0.014302) & (0.002197) & (0.0004) & \\
    & Random grid & 680 & +0.118513 & +0.041180 & +0.009602 & 0.7067 \\
    & & & (0.006473) & (0.002129) & (0.0003) & \\
    \midrule
    \multirow{4}{*}{Llama 3.1 (70B Instruct)} & D-optimal & 680 & +0.590124 & +0.161817 & +0.004846 & 0.9386 \\
    & & & (0.010213) & (0.001882) & (0.0000) & \\
    & Random grid & 661 & +0.564216 & +0.201809 & +0.005139 & 0.6097 \\
    & & & (0.010698) & (0.005752) & (0.0001) & \\
    \bottomrule
  \end{tabular}
  \caption{Structural estimates under D-optimal and random admissible-grid menu samples}
  \label{tab:menu-design-robustness}
\end{table}

Table~\ref{tab:alpha-stability-random-seeds} separates two sources of variation in the position-bias estimate.  The D-optimal design fixes a single space-filling estimation sample, while the random-grid design redraws the same number of admissible portfolio pairs from the feasible mean--variance grid.  Holding the model, prompt, estimator, and grid fixed, variation across random seeds therefore measures how much the estimated position-bias term depends on which feasible menus happen to be sampled rather than on the model's logits alone.

The first three random-design columns report the mean, standard deviation, and range of $\hat{\alpha}$ across seeds 42--51.  These are design-sensitivity summaries, not sampling standard errors: the model responses are deterministic for a fixed prompt and model checkpoint, and the only object being varied is the menu design.  The final two columns summarize whether the position-bias term is approximately constant within the explored grid.  ``Tertile spread'' is the average difference between the largest and smallest tertile means of pair-level bias after sorting menus by predicted choice difficulty; ``Hardness $R^2$'' is the average explanatory power of the linear hardness diagnostic that regresses pair-level bias on predicted difficulty and grid location.

The seed-to-seed dispersion is small within each model, indicating that the random admissible-grid estimate is stable once the design rule is fixed.  Stability across random draws does not imply agreement with the D-optimal design, however.  Qwen 3 (8B) is the clearest exception: its random-grid estimate is tightly concentrated around a negative value even though the D-optimal estimate is positive.  Qwen 3 (14B) also shifts upward under random designs.  By contrast, Ministral, both Llama models, and Gemma 2 produce random-grid means close to their D-optimal estimates.  The tertile-spread and hardness diagnostics show that several models still have state-dependent position bias inside the feasible grid, so the headline $\hat{\alpha}$ should be read as a design-weighted average position effect rather than as evidence that bias is constant over all menu difficulties.

\begin{table}[htbp]
  \centering
  \small
  \resizebox{\textwidth}{!}{
    \begin{tabular}{lrrrrrr}
      \toprule
      Model & D-opt $\hat{\alpha}$ & Random mean $\hat{\alpha}$ & Random SD & Random range & Tertile spread & Hardness $R^2$ \\
      \midrule
      Qwen 3 (14B) & 1.383 & 1.600 & 0.025 & [1.562, 1.637] & 0.610 & 0.305 \\
      Qwen 3 (8B) & 0.305 & -0.987 & 0.035 & [-1.055, -0.922] & 0.367 & 0.274 \\
      Gemma 4 (31B Instruct) & -0.573 & -0.717 & 0.047 & [-0.772, -0.613] & 0.848 & 0.185 \\
      Ministral 3 (8B Instruct) & 0.384 & 0.363 & 0.025 & [0.317, 0.412] & 0.425 & 0.255 \\
      Gemma 2 (9B Instruct) & 0.717 & 0.789 & 0.040 & [0.734, 0.878] & 0.902 & 0.163 \\
      Llama 3.1 (8B Instruct) & 0.172 & 0.124 & 0.005 & [0.119, 0.131] & 0.097 & 0.287 \\
      Llama 3.1 (70B Instruct) & 0.590 & 0.559 & 0.009 & [0.544, 0.571] & 0.068 & 0.214 \\
      \bottomrule
    \end{tabular}
  }
  \caption{Alpha stability across random admissible-grid menu samples, seeds 42--51}
  \label{tab:alpha-stability-random-seeds}
\end{table}

Figure~\ref{fig:alpha-bias-heatmaps} provides the corresponding visual diagnostic.  For each model, the plot pools the random admissible-grid runs from seeds 42--51, computes the paired bias signal for each canonical/mirror menu pair, and averages that signal within fixed bins of average return $\bar{\mu}$ and average risk $\bar{\sigma}$.  Red cells indicate positive first-position/A-label bias, blue cells indicate negative bias, and all panels use the same color scale.  The heatmaps make the state dependence visible: Qwen 3 (14B) has a positive bias throughout the grid that attenuates at higher average risk, Qwen 3 (8B) and Gemma 4 are negative over much of the grid, and the Llama 8B surface is close to flat.  The exercise is descriptive, but it clarifies which parts of the feasible grid are driving the design-weighted $\hat{\alpha}$ values in Table~\ref{tab:alpha-stability-random-seeds}.

\begin{figure}[htbp]
  \centering
  \includegraphics[width=\textwidth]{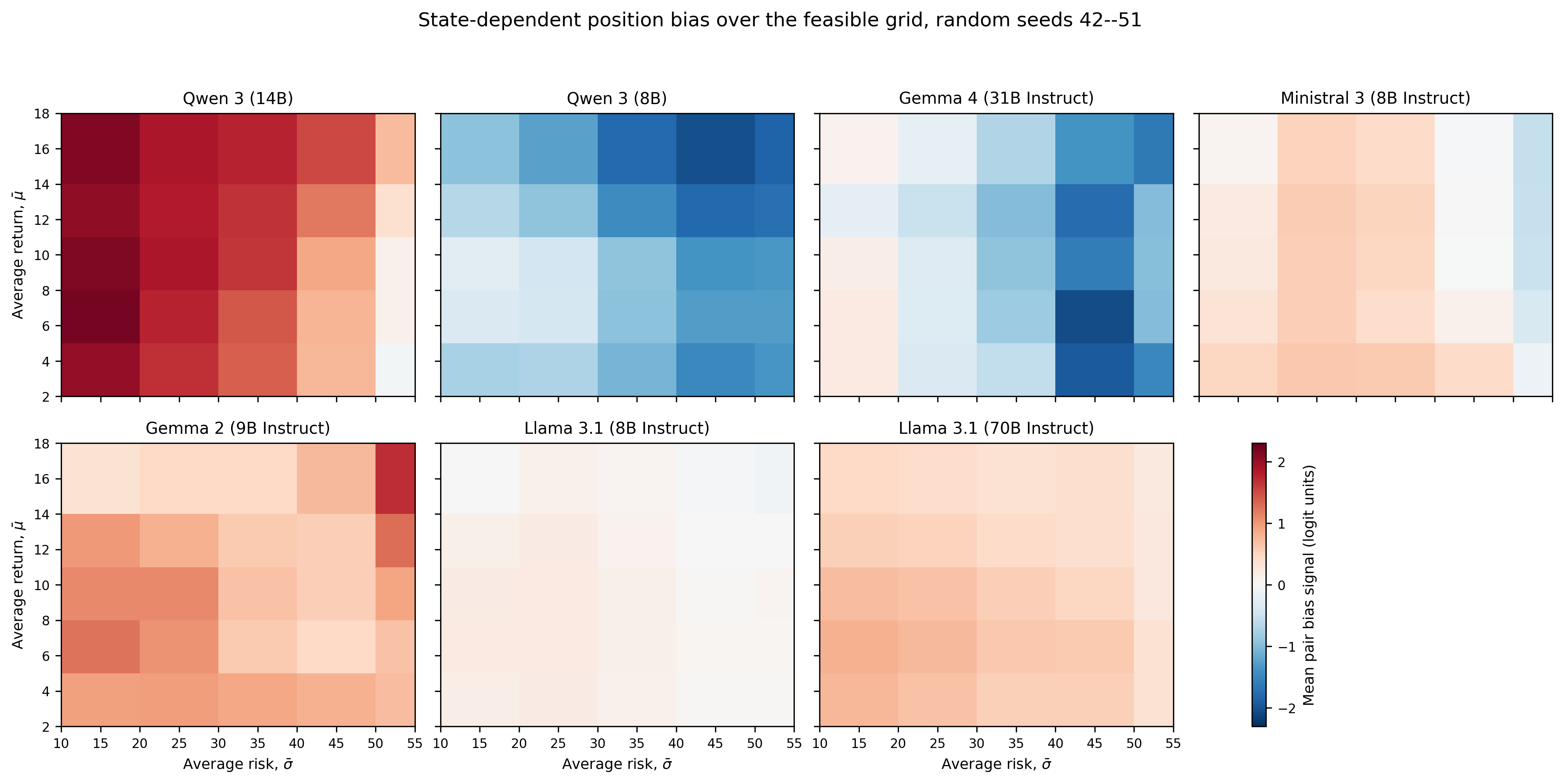}
  \caption{State-dependent position bias over the feasible mean--variance grid. Each cell reports the mean paired bias signal, $0.5(\Delta u_{\text{canonical}}+\Delta u_{\text{mirror}})$, pooled across random admissible-grid menu samples from seeds 42--51.}
  \label{fig:alpha-bias-heatmaps}
\end{figure}

\subsection{Position Bias and Choice Difficulty}

The position-bias term $\hat{\alpha}$ in our structural model is a design-weighted average. This subsection documents that its behavioral footprint is concentrated on menus where the two options are close in utility. We exploit the triad menus from the transitivity design, which are tagged by construction as dominance menus, in which one option Pareto-dominates the other on the mean--variance frontier, as tradeoff menus, in which a higher expected return is purchased with higher risk and the comparison can be resolved only by a risk preference, and as iso-utility menus, in which the two options are equated on the model-implied utility scale. Each menu is presented in a canonical layout and in its position-swapped mirror, so a content-driven chooser must reverse the selected slot across the two layouts. We define the position-determined rate as the share of canonical and mirror pairs in which the model selects the same physical slot in both layouts, that is, the share of comparisons in which the choice is fixed by position rather than by content.

Table~\ref{tab:position-bias-difficulty} reports this rate by model and menu type. On dominance menus every model selects the dominant option with perfect accuracy and the position-determined rate is exactly zero, at zero of $1{,}260$ comparisons pooled across models, with a Wilson 95\% confidence interval of $[0.000, 0.003]$. Position exerts no behavioral influence on these menus because the preference signal is large relative to the position term. The rate rises monotonically as the objective utility gap narrows. It reaches $0.140$ on tradeoff menus and $0.420$ on iso-utility menus, where roughly half of all choices are decided by position alone. A $\chi^2$ test of homogeneity across the three menu types rejects equality decisively, with $\chi^2(2) = 759.2$ and $p < 10^{-160}$. The mean content signal, measured as the absolute position-corrected log-odds $|\delta_z|$, averages $6.1$ on dominance menus, $4.2$ on tradeoff menus, and $1.8$ on iso-utility menus, so the ordering of behavioral position bias mirrors the inverse ordering of preference strength. Position bias is therefore not a uniform labeling artifact. It is latent in the model logits and becomes behaviorally active only on the preference-resolvable tradeoffs that our estimation design is constructed to elicit.

\begin{table}[htbp]
  \centering
  \small
  \begin{tabular}{lccc}
    \toprule
    Model & Dominance & Tradeoff & Iso-utility \\
    \midrule
    Qwen 3 (14B)              & 0.000 & 0.161 & 0.389 \\
    Qwen 3 (8B)               & 0.000 & 0.200 & 0.528 \\
    Gemma 4 (31B Instruct)    & 0.000 & 0.044 & 0.356 \\
    Ministral 3 (8B Instruct) & 0.000 & 0.094 & 0.433 \\
    Gemma 2 (9B Instruct)     & 0.000 & 0.183 & 0.267 \\
    Llama 3.1 (8B Instruct)   & 0.000 & 0.144 & 0.378 \\
    Llama 3.1 (70B Instruct)  & 0.000 & 0.150 & 0.589 \\
    \midrule
    Pooled (all models)       & 0.000 & 0.140 & 0.420 \\
    \bottomrule
  \end{tabular}
  \caption{Position-determined choice rate by menu type. Each cell reports the share of canonical and mirror menu pairs, with $N = 180$ per model and menu type, in which the model selects the same physical slot in both layouts.}
  \label{tab:position-bias-difficulty}
\end{table}

\subsection{Allowing Chain of Thought Reasoning}

To test whether explicit reasoning reduces position bias, we evaluated four models using a Chain of Thought (CoT) prompt \citep{wei2022chain}. The models were instructed to explain their reasoning before outputting a final choice. As shown in Table~\ref{tab:cot-alpha}, CoT reasoning does not systematically eliminate bias; in fact, for most models, the estimated position bias $\hat{\alpha}$ increases significantly or flips direction when CoT is enabled. This suggests that the generated reasoning traces may introduce new forms of path dependence or rationalize the position bias rather than mitigating it.

\begin{table}[htbp]
  \centering
  \small
  \begin{tabular}{lrr}
    \toprule
    Model & D-opt $\hat{\alpha}$ (No CoT) & D-opt $\hat{\alpha}$ (With CoT) \\
    \midrule
    Llama 3.1 (8B Instruct) & 0.172 & 2.628 \\
    Ministral 3 (8B Instruct) & 0.384 & 1.520 \\
    Gemma 2 (9B Instruct) & 0.717 & 2.134 \\
    Qwen 3 (8B) & 0.305 & -1.141 \\
    \bottomrule
  \end{tabular}
  \caption{Position bias $\hat{\alpha}$ with and without Chain of Thought reasoning.}
  \label{tab:cot-alpha}
\end{table}

Table~\ref{tab:cot-params} reports the corresponding structural estimates of the scale parameter $\hat{\kappa}$ and the risk aversion parameter $\hat{\beta}$. We see that explicit reasoning generally increases the scale parameter $\hat{\kappa}$, indicating more deterministic choices, but the effect on the risk aversion parameter $\hat{\beta}$ is mixed.

\begin{table}[htbp]
  \centering
  \small
  \begin{tabular}{lrrrr}
    \toprule
    & \multicolumn{2}{c}{No CoT} & \multicolumn{2}{c}{With CoT} \\
    \cmidrule(lr){2-3} \cmidrule(lr){4-5}
    Model & $\hat{\kappa}$ & $\hat{\beta}$ & $\hat{\kappa}$ & $\hat{\beta}$ \\
    \midrule
    Llama 3.1 (8B Instruct) & 0.038552 & 0.009121 & 0.181967 & 0.005021 \\
    Ministral 3 (8B Instruct) & 0.097941 & 0.011363 & 0.236506 & 0.006451 \\
    Gemma 2 (9B Instruct) & 0.249116 & 0.007871 & 0.259511 & 0.009749 \\
    Qwen 3 (8B) & 0.310677 & 0.003057 & 0.607995 & 0.003874 \\
    \bottomrule
  \end{tabular}
  \caption{Structural estimates of scale $\kappa$ and risk aversion $\beta$ with and without Chain of Thought reasoning.}
  \label{tab:cot-params}
\end{table}

\clearpage
\bibliographystyle{jpe}
\bibliography{REFS}

\end{document}